\documentclass[12pt]{article}

\usepackage{amsmath,amsfonts,amssymb,color,epsfig,graphics,graphicx,latexsym,revsymb,theorem,url,verbatim}

\usepackage{fullpage}

\newtheorem{definition}{Definition}
\newtheorem{proposition}[definition]{Proposition}
\newtheorem{lemma}[definition]{Lemma}

\newtheorem{theorem}[definition]{Theorem}
\newtheorem{corollary}[definition]{Corollary}
\newtheorem{conjecture}[definition]{Conjecture}

\newtheorem{remark}[definition]{Remark}
\newtheorem{example}[definition]{Example}
\newtheorem{question}[definition]{Question}

\def\squareforqed{\hbox{\rlap{$\sqcap$}$\sqcup$}}
\def\qed{\ifmmode\squareforqed\else{\unskip\nobreak\hfil
\penalty50\hskip1em\null\nobreak\hfil\squareforqed
\parfillskip=0pt\finalhyphendemerits=0\endgraf}\fi}
\def\endenv{\ifmmode\;\else{\unskip\nobreak\hfil
\penalty50\hskip1em\null\nobreak\hfil\;
\parfillskip=0pt\finalhyphendemerits=0\endgraf}\fi}
\newenvironment{proof}{\noindent \textbf{{Proof.~} }}{\qed}
\def\Dbar{\leavevmode\lower.6ex\hbox to 0pt
{\hskip-.23ex\accent"16\hss}D}
\makeatletter
\def\url@leostyle{%
  \@ifundefined{selectfont}{\def\UrlFont{\sf}}{\def\UrlFont{\small\ttfamily}}}
\makeatother
\def\bcj{\begin{conjecture}}
\def\ecj{\end{conjecture}}
\def\bcr{\begin{corollary}}
\def\ecr{\end{corollary}}
\def\bd{\begin{definition}}
\def\ed{\end{definition}}
\def\bea{\begin{eqnarray}}
\def\eea{\end{eqnarray}}
\def\bem{\begin{enumerate}}
\def\eem{\end{enumerate}}
\def\bex{\begin{example}}
\def\eex{\end{example}}
\def\bim{\begin{itemize}}
\def\eim{\end{itemize}}
\def\bl{\begin{lemma}}
\def\el{\end{lemma}}
\def\bpf{\begin{proof}}
\def\epf{\end{proof}}
\def\bpp{\begin{proposition}}
\def\epp{\end{proposition}}
\def\bqu{\begin{question}}
\def\equ{\end{question}}
\def\br{\begin{remark}}
\def\er{\end{remark}}
\def\bt{\begin{theorem}}
\def\et{\end{theorem}}

\def\btb{\begin{tabular}}
\def\etb{\end{tabular}}

\newcommand{\nc}{\newcommand}

\def\a{\alpha}
\def\b{\beta}
\def\g{\gamma}
\def\d{\delta}

\def\x{\xi}

\def\r{\rho}
\def\s{\sigma}

\def\ph{\varphi}

\def\ps{\psi}

\def\D{\Delta}
\def\T{\Theta}
\def\L{\Lambda}

\def\S{\Sigma}

 \nc{\bA}{{\bf A}} \nc{\bB}{{\bf B}} \nc{\bC}{{\bf C}}
 \nc{\bD}{{\bf D}} \nc{\bE}{{\bf E}} \nc{\bF}{{\bf F}}
 \nc{\bG}{{\bf G}} \nc{\bH}{{\bf H}} \nc{\bI}{{\bf I}}
 \nc{\bJ}{{\bf J}} \nc{\bK}{{\bf K}} \nc{\bL}{{\bf L}}
 \nc{\bM}{{\bf M}} \nc{\bN}{{\bf N}} \nc{\bO}{{\bf O}}
 \nc{\bP}{{\bf P}} \nc{\bQ}{{\bf Q}} \nc{\bR}{{\bf R}}
 \nc{\bS}{{\bf S}} \nc{\bT}{{\bf T}} \nc{\bU}{{\bf U}}
 \nc{\bV}{{\bf V}} \nc{\bW}{{\bf W}} \nc{\bX}{{\bf X}}
 \nc{\bZ}{{\bf Z}}

\def\bC{{\mathbb{C}}}
\def\bR{{\mathbb{R}}}

\nc{\cA}{{\cal A}} \nc{\cB}{{\cal B}} \nc{\cC}{{\cal C}}
\nc{\cD}{{\cal D}} \nc{\cE}{{\cal E}} \nc{\cF}{{\cal F}}
\nc{\cG}{{\cal G}} \nc{\cH}{{\cal H}} \nc{\cI}{{\cal I}}
\nc{\cJ}{{\cal J}} \nc{\cK}{{\cal K}} \nc{\cL}{{\cal L}}
\nc{\cM}{{\cal M}} \nc{\cN}{{\cal N}} \nc{\cO}{{\cal O}}
\nc{\cP}{{\cal P}} \nc{\cQ}{{\cal Q}} \nc{\cR}{{\cal R}}
\nc{\cS}{{\cal S}} \nc{\cT}{{\cal T}} \nc{\cU}{{\cal U}}
\nc{\cV}{{\cal V}} \nc{\cW}{{\cal W}} \nc{\cX}{{\cal X}}
\nc{\cZ}{{\cal Z}}

\nc{\hA}{{\hat{A}}} \nc{\hB}{{\hat{B}}} \nc{\hC}{{\hat{C}}}
\nc{\hD}{{\hat{D}}} \nc{\hE}{{\hat{E}}} \nc{\hF}{{\hat{F}}}
\nc{\hG}{{\hat{G}}} \nc{\hH}{{\hat{H}}} \nc{\hI}{{\hat{I}}}
\nc{\hJ}{{\hat{J}}} \nc{\hK}{{\hat{K}}} \nc{\hL}{{\hat{L}}}
\nc{\hM}{{\hat{M}}} \nc{\hN}{{\hat{N}}} \nc{\hO}{{\hat{O}}}
\nc{\hP}{{\hat{P}}} \nc{\hR}{{\hat{R}}} \nc{\hS}{{\hat{S}}}
\nc{\hT}{{\hat{T}}} \nc{\hU}{{\hat{U}}} \nc{\hV}{{\hat{V}}}
\nc{\hW}{{\hat{W}}} \nc{\hX}{{\hat{X}}} \nc{\hZ}{{\hat{Z}}}

\def\I{\mathop{\rm I}}

\def\lin{\mathop{\rm span}}

\def\max{\mathop{\rm max}}
\def\min{\mathop{\rm min}}

\def\tr{\mathop{\rm Tr}}

\def\GL{{\mbox{\rm GL}}}

\def\SL{{\mbox{\rm SL}}}

\def\SU{{\mbox{\rm SU}}}

\def\Un{{\mbox{\rm U}}}

\def\es{\emptyset}

\def\ox{\otimes}

\def\sue{\subseteq}

\def\we{\wedge}
\def\bwe{\mathop{\wedge}\nolimits}

\newcommand{\bra}[1]{\langle#1|}
\newcommand{\ket}[1]{|#1\rangle}
\newcommand{\proj}[1]{| #1\rangle\!\langle #1 |}

\newcommand{\braket}[2]{\langle#1|#2\rangle}

\newcommand{\norm}[1]{\lVert#1\rVert}
\newcommand{\abs}[1]{|#1|}

\def\Dbar{\leavevmode\lower.6ex\hbox to 0pt
{\hskip-.23ex\accent"16\hss}D}

\begin{document}
\title{Canonical form of three-fermion pure-states with six single particle states}

\author{Lin Chen${}^{1,2,3}$,
  Dragomir {\v{Z} \Dbar}okovi{\'c}${}^{1,2}$,
  Markus Grassl$^{3}$,
  Bei Zeng${}^{2,4}$\\
\small ${}^{1}$ Department of Pure Mathematics, University of Waterloo, Waterloo, Ontario, Canada\\
\small ${}^{2}$ Institute for Quantum Computing, University of Waterloo, Waterloo, Ontario, Canada\\
\small ${}^{3}$ Centre for Quantum Technologies, National University of Singapore, Singapore\\
\small ${}^{4}$ Department of Mathematics \&  Statistics, University of Guelph, Guelph, Ontario, Canada}

\date{June 6, 2013}

\maketitle

\begin{abstract}
We construct a canonical form for pure states in $\bwe^3(\bC^6)$, the
three-fermion system with six single particle states, under local
unitary (LU) transformations, i.e., the unitary group $\Un(6)$. We
also construct a minimal set of generators of the algebra of
polynomial $\Un(6)$-invariants on $\bwe^3(\bC^6)$. It turns out that
this algebra is isomorphic to the algebra of polynomial LU-invariants
of three-qubits which are additionally invariant under qubit
permutations.  As a consequence of this surprising fact, we deduce that
there is a one-to-one correspondence between the $\Un(6)$-orbits of
pure three-fermion states in $\bwe^3(\bC^6)$ and the LU orbits of pure
three-qubit states when qubit permutations are allowed.
As an important byproduct, we obtain a new canonical form for pure
three-qubit states under LU transformations
$\Un(2)\times\Un(2)\times\Un(2)$ (no qubit permutations allowed).
\end{abstract}

\tableofcontents

\section{Introduction}
\label{sec:intro}

The underlying resource, which makes quantum information processing
more powerful than its classical counterpart, is quantum
entanglement.  The quantum correlations among a number of players can
be much stronger than any classical correlation.  Quantum
entanglement, as a correlation property, should be preserved in
certain ways when only `local' operations are performed.  One of the
central tasks of entanglement theory is to classify entanglement types
with respect to local operations.

More explicitly, for an $N$-particle state $\ket{\psi}$ in the
Hilbert space $\otimes^N(\bC^M)$, a local operation of the form
$\bigotimes_{i=1}^N A_i$ acting on $\ket{\psi}$ preserves certain
entanglement properties of $\ket{\psi}$.  The two most studied cases
are (i) $A_i\in\Un(M)$, i.e. all $A_i$ are unitary, where all the
entanglement properties of $\ket{\psi}$ are preserved; (ii)
$A_i\in\GL(M,\bC)$, where entanglement properties of $\ket{\psi}$
under local operation and classical communication (SLOCC) are
preserved \cite{dvc00,cc06,ccd10}. As throughout the paper we use
the general linear group only over the field of complex numbers, for
simplicity we will write $\GL(M,\bC)$ as $\GL(M)$. Note that
$\ket{\psi}$ is non-entangled if and only if it has the form
$\ket{v_1}\ox\ket{v_2}\ox\cdots\ox\ket{v_N}$ for some
$\ket{v_i}\in\bC^M$.  One of the central tasks in the study of these
$N$-party systems is to classify different states, i.e., the orbits
under these local Lie groups.

Entanglement of identical particle systems (e.g., fermions and bosons)
has also been extensively studied
\cite{AFOV08,SLM01,SCK+01,LZL+01,PY01,ESB+02}, where the concept of
entanglement is explored in a mathematical structure different from
the tensor product structure of the Hilbert space
$\otimes^N(\bC^M)$.  The Hilbert space is now the symmetric power
$\mathop{\vee}\nolimits^N (\bC^M)$ for bosonic systems, and the
exterior power $\bwe^N(\bC^M)$ for fermionic systems, which can be
identified with the subspace of the symmetric and the antisymmetric
tensors in $\otimes^N(\bC^M)$, respectively.

The non-entangled states have the form $\ket{v}^{\vee N}$ in the
bosonic case \cite{PY01,AFOV08,ESB+02}, and the form of a Slater
determinant $\ket{v_1}\wedge\ket{v_2}\wedge\ldots\wedge\ket{v_N}$ in
the fermionic case \cite{LZL+01}. Then an LU operation is an element
$U$ of the unitary group $\Un(M)$ acting on the states as the operator
$U^{\otimes N}$. Similarly, an SLOCC operation $A$ corresponds to an
element of the group $\GL(M)$.

In this work we consider the fermionic system with $3$ fermions with $6$ single particle states, i.e. $N=3$ and $M=6$. We denote
by $\bwe^3(\bC^6)/\Un(6)$ the set of $\Un(6)$-orbits in
$\we^3(V)$, and define similarly the quotient set $\ox^3(\bC^2)/ \left( (\Un(2)\times\Un(2)\times\Un(2))\rtimes S_3 \right)$.
Our first main result is to establish the following natural bijective correspondence:
\begin{equation}
\label{eq:main} \bwe^3(\bC^6)/\Un(6)\leftrightarrow
\ox^3(\bC^2)/ \left( (\Un(2)\times\Un(2)\times\Un(2))\rtimes S_3 \right).
\end{equation}

To be more precise, choose the orthonormal basis vectors of $V=\bC^6$
to be $\ket{i}$, $i=1,\ldots,6$ and set
$e_{ijk}=\ket{i}\we\ket{j}\we\ket{k}$. From these basis vectors we
form three pairs $\{\ket{i},\ket{i+1}\}$, $i=1,3,5$, and introduce
three subspaces
\begin{eqnarray}
\label{eq:Vi} V_1=\lin\{\ket{1},\ket{2}\},\quad
V_2=\lin\{\ket{3},\ket{4}\},\quad V_3=\lin\{\ket{5},\ket{6}\},
\end{eqnarray}
We define the so-called single occupancy vector (SOV) subspace $W$ by
\begin{equation}
 \label{eq:sov}
W:=V_1\we V_2\we V_3\subset\bwe^3(V).
\end{equation}
It was shown in \cite{ccdz13} that any $\ket{\psi}\in \bwe^3(V)$ is LU-equivalent to a state $\ket{\phi}\in W$. The subspace $W$ can be identified with the Hilbert space of three qubits by the isometry $W\to\bC^2\ox\bC^2\ox\bC^2$ which is defined by
\begin{eqnarray}\label{eq:isometry}
e_{i+1,j+3,k+5}\mapsto\ket{ijk},\quad i,j,k\in\{0,1\}.
\end{eqnarray}
This identifies $\ket{\psi}\in W$ with a three-qubit state
$\ket{\phi}$ in a one-to-one manner. The subgroup $G$ of $\Un(6)$
which leaves the subspace $W$ globally invariant is the semidirect product
\begin{eqnarray}
G:=(\Un(2)\times\Un(2)\times\Un(2))\rtimes S_3,
\end{eqnarray}
where the symmetric group $S_3$ permutes the three copies of
$\Un(2)$. Given a $\Un(6)$-orbit $\cO\sue\bwe^3(\bC^6)$, the
intersection $\cO':=\cO\cap W$ is a single $G$-orbit. Conversely, any $G$-orbit in $W$ is contained in a unique $\Un(6)$-orbit
$\cO$. Hence, the correspondence in Eq. \eqref{eq:main} is indeed natural. In terms of quantum entanglement theory, this means that there is a one-to-one correspondence between the LU orbits of
$N=3$, $M=6$ fermionic states and the orbits of three-qubit states under LU transformations and qubit permutations. Moreover
we show that, when the two quotient sets in Eq. \eqref{eq:main}
are equipped with quotient topologies, then our one-to-one correspondence is a homeomorphism. There is no hope to generalize
these results to the case $N\ge4$, $M=2N$ because there exist pure fermionic states in $\bwe^4(\bC^M)$ which are not single occupancy states \cite{ccdz13}.

Our main tool is the powerful invariant theory. In Sec.
\ref{sec:Inv3Ferm} we construct a minimal set of homogeneous
generators of the algebra $\cA$ of polynomial $\Un(6)$-invariants,
which consists of six primary generators, which are algebraically
independent, and an additional secondary generator.  In
Sec. \ref{sec:Inv3Qub} we recall the known results about the algebra
of polynomial LU-invariants of three qubits, and construct a minimal
set of homogeneous generators for its subalgebra $\cB$ consisting of
the invariants that are additionally fixed by qubit permutations.  In
Theorem~\ref{th:equiv} we show that these two algebras are isomorphic
as graded algebras, which leads to the result of Eq.~(\ref{eq:main}).

Based on the full knowledge of the algebra $\cA$, we move further to
construct a canonical form for the fermionic states with $N=3$ and
$M=6$.  This is the second main result of the paper. That is, any state
in $\bwe^3(\bC^6)$ is LU-equivalent to a state
 \begin{equation}
\label{eq:canonical}
\ket{\psi}=a e_{235}+b e_{145}+c e_{136}+d e_{246}+z e_{135},
 \end{equation}
where $a\geq b\geq c\geq 0$, $z=x+iy$ with $x,y$ real, $x\geq 0$, and
$d$ is the maximum of $|{\braket{\a\we\b\we\g}{\ps}}|$ taken over all
unit vectors $\ket{\a\we\b\we\g}$. We denote by $\D$ the spherical
region consisting of all states $\ket{\psi}$ of that form. We show
that $\D$ is a closed connected region and that its relative interior
$\D^0$ is dense in $\D$.  Moreover, we show that if two different
states in $\D$ are LU-equivalent then they must both lie on the
boundary of $\D$, see Theorem \ref{thm:KanForma} for more details.

By using the one-to-one correspondence in Eq.~(\ref{eq:main}) and the
canonical form for the fermionic states, we construct a new canonical
form for three-qubit pure states (without any qubit permutations).
This is the third main result of the paper. The isometry
Eq.~\eqref{eq:isometry} maps the fermionic state
Eq.~\eqref{eq:canonical} to the three-qubit state
\begin{equation}\label{eq:3qubitcanonical}
\ket{\phi}=a\ket{100}+b\ket{010}+c\ket{001}+d\ket{111}+z\ket{000}.
\end{equation}
When $a,b,c,x\geq 0$ and $d= \max_{\alpha,\beta,\gamma}
|{\braket{\alpha,\beta,\gamma}{\ps}}|$, where
$\ket{\alpha,\beta,\gamma}$ is any three-qubit product state,
Eq.~(\ref{eq:3qubitcanonical}) gives a canonical form for three-qubit
pure states. This completes the missing case for three-qubit canonical
forms discussed in \cite{ajt01}, see Sec.~\ref{sec:threequbit} for
more details.

We believe that our results, on both the relationships between the
invariants of Lie groups $\Un(6)$ and $G$, and the canonical forms,
will not only attract interest from quantum information science
community studying entanglement properties, for both
distinguishable particle systems and identical particle systems,
but also will be of interest solely mathematically. Like other
connections between small Lie groups, we certainly believe the
simple format of our results will find applications in yet some other area of
science.

We organize our paper as follows.  In Sec.~\ref{sec:related}, we
review some related work in quantum information theory that studies
local orbits, for both distinguishable and identical particle systems.
We compare these previous works to our results, which further
motivates our work. In Sec.~\ref{sec:RDM}, we discuss reduced density
matrices (RDMs) for pure states of the $N=3$, $M=6$ fermionic system,
and compare them with RDMs of three-qubit systems. The spectra of
these RDMs, which are obviously invariant under $\Un(6)$, will be
later used to build some of the invariants. In
Sec.~\ref{sec:Inv3Ferm}, we discuss the algebra of polynomial $\Un(6)$
invariants of the $N=3$, $M=6$ fermionic system, and in
Sec.~\ref{sec:Inv3Qub} we consider the algebra of the symmetric
polynomial invariants of three qubits.  These two algebras are shown
to be isomorphic as graded algebras. In
Sec.~\ref{sec:KanReg}--\ref{sec:quasireal} we study the canonical form
for pure states of the $N=3$, $M=6$ fermionic system. Finally, in
Sec.~\ref{sec:threequbit} we present a new canonical form for
three-qubit pure states.

\section{Related Work}
\label{sec:related}

\subsection{Systems of distinguishable particles}

It is well-known that in the simplest case of $N=2$ particles, any
state $\ket{\psi}$ in $\bC^M\otimes\bC^M$ is LU-equivalent to a state
in the canonical form given by the Schmidt decomposition
\begin{equation}
\label{eq:Schmidt}
\ket{\psi}=\sum_{i=1}^M\sqrt{\lambda_i}\ket{i}\otimes\ket{i},
\end{equation}
where the states $\ket{i}$, $i=1,\dots,M$ form an orthonormal basis of
$\bC^M$, and $\lambda_1\ge\cdots\ge\lambda_M\ge0$.  Thus states with
different Schmidt coefficients $\lambda_i$ will generically be
LU-inequivalent.  The case $N=3$ turns out to be much more complicated,
as no direct generalization of the Schmidt decomposition is
available.  A canonical form for LU orbits of three qubits was obtained
in \cite{ajt01} as
\begin{equation}
\label{eq:acinform}
\ket{\psi}=\eta_0e^{i\phi}\ket{000}+\eta_1\ket{001}+\eta_2\ket{100}+\eta_3\ket{110}+\eta_4\ket{111},
\end{equation}
where the coefficients $\eta_i$ are real and nonnegative. However, no
canonical form has ever been found for any other $N\geq 3$ system.

The problem of classifying SLOCC orbits was solved, in some other
cases of small systems, thanks to the larger group $\GL(M)^{\times N}$
compared to $\Un(M)^{\times N}$.  It is well-known that three-qubit
pure states can be entangled in two-inequivalent ways \cite{dvc00},
i.e.,
\begin{eqnarray}
&&\ket{GHZ}=\frac{1}{\sqrt{2}}(\ket{000}+\ket{111})\nonumber\\
&\text{and}&\ket{W}=\frac{1}{\sqrt{3}}(\ket{001}+\ket{010}+\ket{100}).
\end{eqnarray}
Four-qubit SLOCC equivalence classes have also been identified,
however, now infinitely many orbits exist \cite{VDMV12,CD06}. SLOCC
equivalence classes for other cases have also been discussed, for
instance the $2\times M\times N$ states investigated by the range
criterion \cite{cc06}, multipartite symmetric states by the tensor
rank \cite{ccd10}, and multiqubit symmetric states by locally identical
operators \cite{mkg10}. In all these investigations invariant theory
plays an important role.

Recently, a connection of SLOCC classification with the theory of the
quantum marginal problem has been studied
\cite{WDGC12,SOK12a,SOK12b}. It was shown that the set of vectors with
entries given by the eigenvalues of the one particle reduced density
matrices is a convex polytope (namely the entanglement polytope), for
the set of states in any SLOCC equivalence class. Further, it was
shown that there is only a finite number of polytopes for any $N$ and
$M$, compared to an in general infinite number of SLOCC orbits.
Hence, the result in \cite{WDGC12,SOK12a,SOK12b} provides a
coarse-grained version of the SLOCC classification.  In identifying
these polytopes, invariant theory also plays a crucial role.

\subsection{Fermionic systems}

In the case $N=2$ and even local dimension $M=2K$, canonical forms for
LU orbits similar to the Schmidt decomposition have been obtained for
both bosonic and fermionic systems \cite{SLM01,LZL+01,PY01}. The
bosonic state is LU-equivalent to a state which has exactly the same
form as Eq.~(\ref{eq:Schmidt}). It means that we have a one-to-one
correspondence
$\vee^2(\bC^M)/\Un(M)\leftrightarrow(\bC^M\otimes\bC^M)/(\Un(M)\times
\Un(M))$. The fermionic case is a bit different as any fermionic state
is LU-equivalent to a special form of Eq.~(\ref{eq:Schmidt})
\begin{equation}
\ket{\psi}=\sum_{i=1}^K \sqrt{\lambda_i}\ket{\alpha_i}\wedge\ket{\beta_i},
\end{equation}
where $\bra{\alpha_i}\alpha_j\rangle=\bra{\beta_i}\beta_j\rangle=\delta_{ij}$, and
$\bra{\alpha_i}\beta_j\rangle=0$. This means that we have a
one-to-one correspondence
\begin{equation}
\wedge^2 (\bC^M)/\Un(M)\leftrightarrow(\bC^K\otimes\bC^K)/(\Un(K)\times\Un(K)).
\end{equation}
This fact is indeed known, see e.g., \cite{ccdz13}.

For the case $N\geq 3$, there is no generalization of the Schmidt
decomposition. So far no nontrivial canonical forms of LU orbits for
$N\geq 3$ bosonic/fermionic systems have been identified.  By
``non-trivial" we mean that there are some trivial cases which can be
treated easily.  For instance, for bosonic systems with $M=2$, the LU
group is just $\Un(2)$.  For fermionic systems, due to the
particle-hole duality, only the cases $M\geq 2N$ are of interest.

The simplest nontrivial fermionic system with $N=3$ and $M=6$, i.e.,
three fermions with six single particle states, has attracted much
attention recently.  Its SLOCC orbits have been completely
classified, and it turns out that there is a surprising link to the
SLOCC orbits of three-qubits \cite{LV08}. That is, there is a
one-to-one correspondence between the SLOCC orbits of these two very
different systems
\begin{equation}
\label{eq:SLOCCorbit}
\wedge^3 (\bC^6)/\GL(6)\leftrightarrow (\bC^2\ox\bC^2\ox\bC^2)/
\left( (\GL(2)\times\GL(2)\times\GL(2))\rtimes S_3 \right),
\end{equation}
which is yet another accidental coincidence involving small Lie
groups.  For several $N$ and $M$ fermionic systems with $N\ge3$, the
SLOCC orbits in $\we^N(\bC^M)$ were classified as early as 1931 in the
context of multilinear algebra \cite{Sch31}.  The SLOCC orbits and
polynomial invariants of three qubits were studied in detail in 1999
\cite{Ehr99}, in the context of $2\times2\times2$ complex
matrices. The fact that there is a natural correspondence between
these orbits and the SLOCC orbits of the fermionic system
$\we^3(\bC^6)$ is also pointed out in the same reference.

Studies of the $N$-representability problem \cite{Col63} suggest some
further connection of the orbits of the $N=3$, $M=6$ fermionic system
with the orbits of three-qubit system, but in a more complicated
situation where the LU orbits are considered.  Because of the smaller
group $\Un(6)$ compared to $\GL(6)$, one needs to deal with many
more (in fact infinitely many) orbits.  It was shown that, if one
arranges the eigenvalues $\lambda_i$ of the one-particle RDM of any
pure $N=3$, $M=6$ fermionic state $\ket{\psi}$ in a non-increasing order as
$\lambda_1\ge\lambda_2\ge\cdots\ge\lambda_5\ge\lambda_6$, then
$\lambda_i+\lambda_{7-i}=1$ for $i=1,2,3$
\cite{ccdz13,Kly05,BD72,Rus07,Kly09}. This indicates that there is
always a representative of each LU orbit which adopts a special form,
namely the single-occupancy form.

However, one key question remained unanswered: for a given $N=3$,
$M=6$ fermionic state $\ket{\psi}$, can any two of its LU-equivalent
states in $W$ correspond to two LU-inequivalent three-qubit states under
the LU group $\Un(2)\times\Un(2)\times\Un(2)$ and the permutation of
qubits? In other words, is
there a one-to-one correspondence between the LU orbits of $N=3$,
$M=6$ fermionic states and the three-qubit states, just like in the
SLOCC case in \eqref{eq:SLOCCorbit}? This question is answered
affirmatively by \eqref{eq:main} and Theorem \ref{th:equiv} in
Sec. \ref{sec:Inv3Qub}. Although \eqref{eq:main} and
\eqref{eq:SLOCCorbit} are similar, we emphasize that the former
relation is much more important because many physical properties of
three-qubit pure states are invariant under the group
$\Un(2)\times\Un(2)\times\Un(2)$, but not under
$\GL(2)\times\GL(2)\times\GL(2)$. Such properties include entanglement
measures (e.g., geometric measure of entanglement in
Sec. \ref{subsec:GME}), eigenvalues of RDMs and so on. Moreover, the
group $\Un(2)\times\Un(2)\times\Un(2)$ is realizable with probability
one in experiment, while $\GL(2)\times\GL(2)\times\GL(2)$
corresponding to SLOCC can be realized only with a non-vanishing
probability.

Note that when talking about three-qubit orbits, one has to take into account
the qubit permutations by the symmetric group $S_3$. This is
because, in the fermionic case, qubit permutations correspond to
permuting the subspaces $V_1$ ,$V_2$, $V_3$, which preserve $W$. On
the other hand, in considering the three-qubit LU orbits such as
Eq.~(\ref{eq:acinform}), the qubit permutations were not taken into
account \cite{ajt01}. In fact, two states given by
Eq.~(\ref{eq:acinform}) with different parameters $a$, $b$, $c$ may
correspond to the same orbit under
$(\Un(2)\times\Un(2)\times\Un(2))\rtimes S_3$. So
Eq.~(\ref{eq:acinform}) is no longer a canonical form in this case. As
permutation of qubits does not change the main entanglement properties of
quantum states, it should be taken into account for any task involving
entanglement classification.

\subsection{Geometric measure of entanglement} \label{subsec:GME}

In quantum information science the geometric measure of entanglement
(GME) is a known multipartite entanglement measure \cite{shimony95}.
For $N$-partite pure states $\ket{\ps}\in\otimes^N_{i=1}\cH_i$, the
GME is defined as \cite{wg03}
\begin{eqnarray}
 \label{eq:GME}
 G(\ps) = 1 - \L^2 (\ps),
\end{eqnarray}
where $\L(\ps)=\max_{\ket{a_1,\dots,a_N}}
\abs{\braket{a_1,\dots,a_N}{\ps}}$, and $\ket{a_i}\in\cH_i$ are
unit vectors. For fermionic system, we can similarly define the GME
for $N$-vectors as
\begin{eqnarray}
 \label{eq:GMEfermi}
 G_f(\ps) = 1 -\L_f^2 (\ps),
\end{eqnarray}
where $\L_f(\ps)=\max_{\ket{b_1\we\dots\we b_N}}
\abs{\braket{b_1\we\dots\we b_N}{\ps}}$, and $\ket{b_1\we\cdots\we
b_N}$ are unit vectors. In the case $N=3$, $M=6$, for any pure
fermionic state $\ket{\ps}$ we will define the quantity $\mu(\psi)$
to be the maximum overlap $|\braket{\a\we\b\we\g}{\psi}|$ over all
decomposable three-vectors $\ket{\a\we\b\we\g}$ of unit norm in
Definition \ref{def:mu}. In particular, $\L_f(\ps)=\mu(\ps)$ for any
$3$-vector $\ket{\ps}\in W$, where $W$ is the SOV space defined in
\eqref{eq:sov}.  By using the embedding relation \eqref{ea:embedding}
and Lemma \ref{le:maximum}, we can obtain $\L(\ps)=6^{-1/2}\L_f(\ps)$
for $\ket{\ps}\in W$.

Since $\mu$ does not change under local unitary operations, it follows
that if $\mu(\phi)\ne\mu(\psi)$ then $\ket{\phi}$ and $\ket{\psi}$ are
not LU-equivalent.  Recall from Eq. \eqref{eq:main} that when $N=3$,
$M=6$ there is a one-to-one correspondence between the LU orbits of
fermionic states and the LU plus permutations orbits of three-qubit
states. By \eqref{eq:GME}, \eqref{eq:GMEfermi}, and Lemma
\ref{le:maximum}, computing the GME of $\ket{\ps}\in\bwe^3(C^6)$ is
equivalent to the same task for three-qubit pure states.

In recent years, the GME has been used in the context of many
different aspects in quantum information.  First, most multipartite
states have been shown to be too entangled to implement
measurement-based quantum computing {\em via} GME \cite{gfe09}. GME
is also important for studying multiple-prover quantum Merlin-Arthur
games \cite{hm10} and weak multiplicativity for random quantum
channels \cite{montanaro13}.  Second, there have been efforts to
search for the maximally entangled multiqubit states with respect to
the GME \cite{cxz10,amm10}. Research has also been carried out over
the additivity, and computation of GME for Dicke, Smolin,
stabilizer, and antisymmetric states \cite{hmm08,zch10}.  Third,
computing the GME is helpful for the understanding of other
important entanglement measures including the entanglement of
formation \cite{wg03}, relative entropy of entanglement, and
entanglement of robustness \cite{weg04,zch10}. Fourth, the GME is
also related to the long-standing open problems on
symmetrical-informationally-complete (SIC)-POVM and mutually
unbiased bases (MUB) \cite{czw10}.

Thus studying the GME is important, however it has been proved to be
hard even for three-qubit pure states \cite{cxz10}. Our Lemma
\ref{le:connection} gives a general method for computing the
quantity $\mu$ (and the GME) of the linear superposition of two
states with known values of $\mu$.  Next, Theorem \ref{thm:NejDelta}
gives a collection of inequalities which cut out a spherical region
$\D$, by which one can decide whether a given positive number is the
value of $\mu$ of some three-qubit pure state.  Third, we show that
the maximally entangled three-qubit pure state with respect to the
GME is the $W$ state, see Proposition \ref{pp:Minimum}.  This gives
an independent proof of the main result of \cite{cxz10}. Fourth,
Proposition \ref{pp:Delta} shows that the region $\D$ provides the
canonical form for three-fermionic states $\ket{\ps}$, in which
$d=\mu(\ps)$ is one of the coordinates.  Fifth, we give an algorithm
to compute the GME for three-qubit states in Sec.
\ref{subsec:algorithm}.

\section{Reduced Density Matrices}
\label{sec:RDM}

We consider $N$-vectors of an $M$-dimensional Hilbert space $V$ as
antisymmetric tensors.  We shall use the isometric embedding
\begin{eqnarray}
 \label{ea:embedding}
\ket{v_1\we v_2\we\cdots\we v_N}\to\frac{1}{\sqrt{N!}}
\sum_{\s\in S_N} {\rm sgn}(\s)
\ket{v_{\s(1)},v_{\s(2)},\ldots,v_{\s(N)}}.
\end{eqnarray}
Thus we may consider $\we^N(V)$ as a subspace of $\otimes^N(V)$.
Consequently, we may apply the partial trace operators to
the density matrix $\r=\proj{\ps}$.  Although $\r$ acts on
$\otimes^N(V)$, its support (i.e., range) is contained in
$\we^N(V)$ and we shall identify $\r$ with its restriction to
$\we^N(V)$.

The general linear group $\GL:=\GL(V)$ acts on $\cH$ by the
so called {\em diagonal action}:
\begin{align*}
A\cdot(\ket{v_1}\otimes\ket{v_2}\otimes
\cdots\otimes\ket{v_N})=
A\ket{v_1}\otimes A\ket{v_2}\otimes
\cdots\otimes A\ket{v_N}, \quad A\in\GL,
\quad \ket{v_i}\in V.
\end{align*}
In other words, $A\in\GL$ acts on $\cH$ as $\otimes^N
A$. Similarly, we have the action of $\GL$ on $\wedge^N V$ where
$A\in\GL$ acts as $\wedge^N A$, i.e., we have

\begin{eqnarray}
A\cdot\ket{v_1\wedge\cdots\wedge v_N} &=&
\ket{Av_1\wedge\cdots\wedge Av_N}.
\end{eqnarray}

These actions can be restricted to the unitary group $\Un(V)$ of
$V$. We shall say that two $N$-vectors $\ket{\phi}$ and $\ket{\psi}$
are {\em equivalent} if they belong to the same $\GL$-orbit, i.e.,
$\ket{\psi}=A\cdot\ket{\phi}$ for some $A\in\GL$. We shall also say
that they are {\em unitarily equivalent} or {\em LU-equivalent} if $A$
can be chosen to be unitary.

In this paper we consider only the case where $M=6$ and $N=3$, that
is, three fermions with six single particle states. We have the
following result from \cite{ccdz13}.
\begin{lemma}
\label{lm:5basis} Any three-fermion pure state with six single
particle states is LU-equivalent to
\begin{equation} \label{eq:5dimVektor}
\ket{\psi}=a e_{235}+b e_{145}+c e_{136}+d e_{246}+z e_{135},
\end{equation}
where $e_{ijk}=\ket{i}\we\ket{j}\we\ket{k}$, the coefficients
$a,b,c,d\ge0$, $z\in\bC$, and
$\|\psi\|^2=a^2+b^2+c^2+d^2+|z|^2=1$.
\end{lemma}

The symmetric group $S_3$ is generated by the 3-cycle $\s=(123)$ and
the transposition $\tau=(12)$. We embed $S_3$ into the unitary group
$\Un(6)$ as follows:
\begin{eqnarray} \label{eq:embed}
\s=\left( \begin{array}{ccc}0&0&I_2\\I_2&0&0\\0&I_2&0
\end{array} \right),\quad
\tau=-\left( \begin{array}{ccc}0&I_2&0\\I_2&0&0\\0&0&I_2
\end{array} \right),
\end{eqnarray}
where $I_2$ is the $2\times 2$ identity matrix.

Note that $S_3$ permutes the subspaces $V_1$, $V_2$,
$V_3$. Consequently, under the action of $\Un(6)$ on $\we^3(V)$, $S_3$
leaves invariant the subspace $W$. Moreover, $S_3$ permutes the basis
vectors $e_{135}$, $e_{145}$, $e_{136}$, $e_{146}$, $e_{235}$,
$e_{245}$, $e_{236}$ ,$e_{246}$ of $W$. The $3$-cycle $\s$ fixes
$e_{135}$ and $e_{246}$, and sends $e_{145}\to e_{136}\to e_{235}\to
e_{145}$ and $e_{146}\to e_{236}\to e_{245}\to e_{146}$. The
transposition $\tau$ fixes $e_{135}$, $e_{136}$, $e_{245},e_{246}$,
and sends $e_{145}\to e_{235}\to e_{145}$ and $e_{146}\to e_{236}\to
e_{146}$. It follows that under this action, $S_3$ preserves the real
$6$-dimensional subspace $W_6$ with basis
\begin{eqnarray} \label{eq:BazW6}
\{e_{235},e_{145},e_{136},e_{246},e_{135},ie_{135}\}.
\end{eqnarray}
In terms of the coordinates $a,b,c,d,z$, the action of $\s$ and $\tau$
on $W_6$ is given by
 \begin{eqnarray}
 \label{eq:sigmatau}
\s\cdot(a,b,c,d,z)=(c,a,b,d,z)\quad\text{and}\quad
\tau\cdot(a,b,c,d,z)=(b,a,c,d,z).
\end{eqnarray}
(This is the reason why we introduced the factor $-1$ in the
definition of the matrix of $\tau$.) Thus we can assume that
$a\ge b\ge c\ge0$ and $d\ge0$.

To compute the $1$-body RDM for $\ket{\psi}$, we use a different
normalization i.e. it has trace $3$, as used by chemists
\cite{Kly05}. Thus we set $\r^{(1)}=3\r_1$, where $\r_1$ is the
normalized $1$-RDM.

A computation shows that $\r^{(1)}=R_a\oplus R_b\oplus R_c$,
where
\begin{eqnarray*}
R_a &=&
\left(\begin{array}{cc}
b^2+c^2+|z|^2 & az \\
az^* & a^2+d^2
\end{array}\right), \\
R_b &=&
\left(\begin{array}{cc}
c^2+a^2+|z|^2 & bz \\
bz^* & b^2+d^2
\end{array}\right), \\
R_c &=&
\left(\begin{array}{cc}
a^2+b^2+|z|^2 & cz \\
cz^* & c^2+d^2
\end{array}\right).
\end{eqnarray*}
For $x=a,b,c$ let $D_x=\det R_x$, i.e.,
\begin{eqnarray}
D_a&=&(b^2+c^2)(a^2+d^2)+d^2|z|^2, \label{eq:Da} \\
D_b&=&(c^2+a^2)(b^2+d^2)+d^2|z|^2, \label{eq:Db} \\
D_c&=&(a^2+b^2)(c^2+d^2)+d^2|z|^2. \label{eq:Dc}
\end{eqnarray}

Since $\tr R_x=1$ and $R_x\ge0$, we have $D_x\in[0,1/4]$, and the
eigenvalues of $R_x$ can be written as $\lambda_x$ and
$1-\lambda_x$ with $\lambda_x=(1+\sqrt{1-4D_x})/2\in[1/2,1]$.
Let us denote the eigenvalues of $\r^{(1)}$ arranged in
decreasing order as
$\lambda_1\ge\lambda_2\ge\cdots\ge\lambda_5\ge\lambda_6$.
Then $\lambda_i$ and $\lambda_{7-i}$ are the eigenvalues of the
same block $R_x$ of $\r^{(1)}$. Thus we
obtain that if the $\lambda_i$ are arranged in decreasing order,
then $\lambda_i+\lambda_{7-i}=1$ for $i=1,2,3$ (this fact is proved
in \cite{Col63,Kly05,Rus07,Kly09} using other methods).

Next we consider the $2$-particle RDM $\rho_{1,2}$. This is a
$\binom{6}{2}\times\binom{6}{2}$ matrix, which is also block
diagonal with three $4\times 4$ blocks,
\begin{description}
\item[Block 1:] coordinates $(1,3)$, $(1,4)$, $(2,3)$ ,$(2,4)$, the corresponding block in
  $\rho_{1,2}$ is denoted by $\rho_{1,2}^{[1]}$,
\item[Block 2:] coordinates $(1,5)$, $(1,6)$, $(2,5)$, $(2,6)$, the corresponding block in
  $\rho_{1,2}$ is denoted by $\rho_{1,2}^{[2]}$,
\item[Block 3:] coordinates $(3,5)$, $(3,6)$, $(4,5)$, $(4,6)$, the corresponding block in
  $\rho_{1,2}$ is denoted by $\rho_{1,2}^{[3]}$,
\end{description}
and a zero $3\times 3$ block at the coordinates $(1,2)$, $(3,4)$, $(5,6)$.

Using Eq. (\ref{eq:5dimVektor}), the three non-zero blocks are given by
\[
\rho_{1,2}^{[1]}=
\left(\begin{array}{cccc}
c^2+|z|^2 & bz & az & cd \\
bz^* & b^2 & ab & 0  \\
az^* & ab & a^2 & 0 \\
cd & 0 & 0 & d^2
\end{array}\right),
\]
\[
\rho_{1,2}^{[2]}=
\left(\begin{array}{cccc}
b^2+|z|^2 & cz & az & bd \\
cz^* & c^2 & ac & 0  \\
az^* & ac & a^2 & 0 \\
bd & 0 & 0 & d^2
\end{array}\right),
\]
\[
\rho_{1,2}^{[3]}=
\left(\begin{array}{cccc}
a^2+|z|^2 & cz & bz & ad \\
cz^* & c^2 & bc & 0  \\
bz^* & bc & b^2 & 0 \\
ad & 0 & 0 & d^2
\end{array}\right).
\]

Note that for $\ket{\psi}$ given in Eq. (\ref{eq:5dimVektor}),
$\ket{\psi}\in W$ is a SOV state. So $\ket{\psi}$ can be identified as
a qubit state for qubits $A$, $B$, $C$, via the map
$e_{i+1,j+3,k+5}\to\ket{ijk},\ i,j,k\in\{0,1\}$, i.e.
\begin{equation} \label{eq:5dimVektorqubits}
\ket{\psi}=a\ket{100}+b\ket{010}+c\ket{001}+d\ket{111}+z\ket{000}.
\end{equation}
Then we have
\begin{eqnarray}
\rho_{1,2}^{[1]}&=&\rho_{AB}={\tr}_C\ket{\psi}\bra{\psi},\nonumber\\
\rho_{1,2}^{[2]}&=&\rho_{AC}={\tr}_B\ket{\psi}\bra{\psi},\nonumber\\
\rho_{1,2}^{[3]}&=&\rho_{BC}={\tr}_A\ket{\psi}\bra{\psi}.
\end{eqnarray}

\section{Polynomial $\Un(6)$-invariants of three fermions}
\label{sec:Inv3Ferm}
We continue to use the real $6$-dimensional subspace $W_6\subseteq
\we^3(V)$ introduced in Section \ref{sec:RDM}.  We shall write an
arbitrary $\ket{\psi}\in W_6$ as in Eq. \eqref{eq:5dimVektor}, and use
$a,b,c,d$ and $z=x+iy$ as its coordinates (with $x$ and $y$ real).

We denote by $\Un(6)$ the unitary group of $V$. We set
$\r=\proj{\psi}$ where
\begin{eqnarray} \label{eq:Psi} \ket{\psi}=\sum_{1\le
i<j<k\le 6} \xi_{ijk}e_{ijk}, \quad  \xi_{ijk}\in\bC,
\end{eqnarray}
is an arbitrary $3$-vector. Since we work with non-normalized pure
states we should specify how we normalize the density matrix $\r$ and
its one- and two-body RDMs $\r_1:={{\tr}}_{2,3}(\r)$ and
$\r_{1,2}:={{\tr}}_3(\r)$. We shall require that
\begin{eqnarray} \label{eq:normalization}
\tr\r=\tr\r_1=\tr\r_{1,2}=\|\ps\|^2=\sum_{i<j<k} |\xi_{ijk}|^2.
\end{eqnarray}
Thus, if $\ket{\ps}$ is a unit vector then the density matrices
$\r$, $\r_1$, and $\r_{1,2}$ all have trace $1$.
Note that this normalization is different from the one used in
Section \ref{sec:RDM}.

Let us introduce the $\Un(6)$ or LU-invariants of $\we^3(V)$.  To
begin, we have to consider $\we^3(V)$ as a real vector space, and our
invariants will belong to the algebra $\cP_\bR$ of complex-valued
polynomials defined on this real vector space.  We shall denote this
algebra of invariants by $\cP_\bR^\Un$.  After complexification, we
obtain the representation of $\GL(6)$ on
$\bC\ox_\bR\we^3(V)$. This tensor product decomposes into a direct sum
of two complex holomorphic representations of $\GL(6)$, namely
$\we^3(V)$ and its complex conjugate representation. However, since
$\Un(6)$ is a compact group, the complex conjugate representation is
isomorphic to the dual representation.  These two irreducible
representations of $\GL(6)$ are not isomorphic, but they become
isomorphic when restricted to $\SL(6)$ (again, by $\SL(6)$
we mean $\SL(6,\bC)$). This means that
$\bC\ox_\bR\we^3(V)$ is the direct sum of two copies of the
fundamental representation of $\SL(6)$ on $\we^3(V)$.  Let us
denote by $\cP$ the algebra of complex holomorphic polynomials on
$\bC\ox_\bR\we^3(V)$. The complexification of $\cP_\bR^\SU$ coincides
with the algebra $\cP^\SL$ of holomorphic polynomial
$\SL(6)$-invariants on $\bC\ox_\bR\we^3(V)$.  It is known
\cite{GS78} that this algebra is regular, i.e., it is a polynomial
algebra (in $7$ variables).

Due to the direct decomposition of $\bC\ox_\bR\we^3(V)$, the algebra
$\cP$ is bigraded, and the same is true for its subalgebra
$\cP^\SL$. From \cite[Table 2a]{GS78} we know that the $7$
bihomogeneous generators of $\cP^\SL$ have bidegrees $(1,1)$, $(4,0)$,
$(3,1)$, $(2,2)$, $(1,3)$, $(0,4)$, $(3,3)$.  Hence, the bigraded
Poincar\'{e} series of $\cP^\SL$ is
\begin{eqnarray}
f(s,t)=\frac{1}{(1-st)(1-s^4)(1-s^3t)(1-s^2t^2)(1-st^3)(1-t^4)
(1-s^3t^3)}.
\end{eqnarray}

It is not hard to compute the bigraded Poincar\'{e} series of
$\cP^\GL$.  We can distinguish the two irreducible submodules of
$\bC\ox_\bR\we^3(V)$ by taking into account the action of the
central torus of $\GL(6)$. Thus, we replace $s$ ,$t$ with
$sz$, $tz^{-1}$ respectively, and expand $f(sz,tz^{-1})$ into the
Laurent series with respect to the variable $z$. Then the bigraded
Poincar\'{e} series of $\cP^\GL$, say $g(s,t)$, is the constant
term of that Laurent series. It is given by the contour integral
over the unit circle:
\begin{eqnarray}
g(s,t)=\frac{1}{2\pi i}\int f(sz,tz^{-1})\frac{{\rm d}z}{z}.
\end{eqnarray}

By computing this contour integral by standard methods
(e.g., the Residue Theorem) and then setting $s=t$, we obtain the
ordinary Poincar\'{e} series of $\cP^\GL$:
\begin{eqnarray} \label{eq:serUn}
g(t,t) &=&
\frac{1+t^{12}}{(1-t^2)(1-t^4)(1-t^6)(1-t^8)^2(1-t^{12})} \\
&=& 1+t^2+2t^4+3t^6+6t^8+7t^{10}+13t^{12}+16t^{14}+25t^{16}
+31t^{18}+46t^{20}+\cdots. \notag
\end{eqnarray}
This is also the Poincar\'{e} series of $\cP^\Un_\bR$.

Any polynomial LU-invariant $f$ is uniquely determined by its
restriction $f':=f|_{W_6}$. These restrictions are polynomials in only
$6$ real variables, and so calculations with them are much faster than
with the full expressions for the invariants which depend on $40$ real
variables.  When convenient, we may specify $f$ by giving explicit
expression for its restriction $f'$.

Let us now describe the generators of $\cP_\bR^\Un$ that we shall
use. The norm square, $M_1:=\|\psi\|^2$, is the unique invariant of
degree $2$. By using our normalizations \eqref{eq:normalization}, we
also have $M_1=\tr(\r)=\tr(\r_1)=\tr(\r_{1,2})$.

The representation space $\we^3(\bC^6)$ of $\GL(6)$ is a well
known regular prehomogeneous vector space, \cite[Example 2.5]{TK03}.
We shall denote its relative invariant by $F$. In
particular, $F$ is also an $\SU(6)$-invariant, in fact the unique
invariant of bidegree $(4,0)$, and its complex conjugate is the unique
invariant of bidegree $(0,4)$.  The invariant $F$ was recently
rediscovered by Levay and Vrana who also found a nice explicit formula
for it \cite[Eq. (14)]{LV08}.  When restricted to $W_6$, it has a very
simple expression
\begin{eqnarray} \label{eq:InvF}
F'=d(4abc+dz^2).
\end{eqnarray}

The elementary symmetric functions of the determinants $D_x$ defined
in Eqs. \eqref{eq:Da}--\eqref{eq:Dc}, namely
\begin{eqnarray}
M'_2 &=& D_a+D_b+D_c, \label{eq:D-2} \\
M'_4&=& D_aD_b+D_bD_c+D_cD_a, \label{eq:D-4} \\
M'_6&=& D_aD_bD_c, \label{eq:D-6}
\end{eqnarray}
turn out to be restrictions to $W_6$ of polynomial
$\Un(6)$-invariants which we will denote by $M_2$, $M_4$,
and $M_6$, respectively. Explicitly, these three invariants
can be defined as follows:
\begin{eqnarray}
\label{eq:Rel-1}
M_2 &=& \frac{3}{2} \left( M_1^2-3\tr(\r_1^2) \right), \\
\label{eq:Rel-2}
M_4 &=& \frac{1}{4} \left( 3M_1^4 +2M_2^2 -4M_1^2 M_2
-81\tr(\r_1^4) \right), \\
\label{eq:Rel-3}
M_6  &=& \frac{1}{6} \left( 3M_1^6 -6M_1^4 M_2 +9M_1^2 M_2^2
-18M_1^2 M_4 -2M_2^3 +6M_2 M_4 -729\tr(\r_1^6) \right).
\end{eqnarray}

Next we introduce the invariants $M_3$ and $M_5$ by giving two
different expressions for each:
\begin{eqnarray} \label{eq:M3p}
M_3 &=& \frac{3}{2}M_1M_2 -\frac{1}{8}\|\nabla M_2\|^2 \\
    &=& 3M_1(M_1^2-M_2)-27\tr \left( \r_1{{\tr}}_2(\r_{1,2}^2)
\right), \label{eq:M3d} \\
M_5 &=& \frac{1}{18}( 10M_2^2 +8M_1M_3 -24M_4
-\nabla M_2\cdot\nabla M_3 ), \label{eq:M5nab} \\
    &=& |F|^2, \label{eq:M5F}
\end{eqnarray}
where $\nabla M_i$ denotes the gradient of $M_i$ considered as a
function of $40$ real variables, namely the real and imaginary parts
of the $\xi_{ijk}$ in Eq. \eqref{eq:Psi}.

We shall see below that the invariants $M_1,\ldots,M_6$ defined above
are the primary invariants of $\cP_\bR^\Un$.  There is only one
secondary invariant $M_7$ which will be temporarily specified by its
restriction $M'_7$.

For convenience, we introduce the following abbreviations
\begin{eqnarray} \label{eq:si}
s_1:=a^2+b^2+c^2,\quad
s_2:=a^2b^2+a^2c^2+b^2c^2,\quad
s_3:=a^2b^2c^2.
\end{eqnarray}
Let us now write explicitly the restriction $M'_i$ of the above
seven invariants:
\begin{eqnarray}
M'_1 &=& s_1+d^2+|z|^2, \label{eq:M'1} \\
M'_2 &=& 2(s_2+s_1d^2) +3d^2|z|^2, \label{eq:M'2} \\
M'_3 &=& M'_1(s_2+s_1d^2) -6(s_3+abcd(x^2-y^2)+s_2d^2),
\label{eq:M'3}  \\
M'_4 &=& s_1s_3+s_2^2 +3(s_1s_2-s_3)d^2 +4(s_2+s_1d^2)d^2|z|^2
        +(s_1^2+s_2+3|z|^4)d^4, \label{eq:M'4} \\
M'_5 &=& d^2|4abc+dz^2|^2=d^2(16s_3+8abcd(x^2-y^2)+d^2|z|^4),
\label{eq:M'5} \\
M'_6 &=& (s_1s_2-s_3)(d^6+s_1d^4+s_2d^2+s_3+3d^4|z|^2)
         +(s_1s_3+s_2^2)d^2|z|^2 \notag \\
&& +2s_2d^4|z|^4 +(s_1^2+s_2+2s_1|z|^2+|z|^4)d^6|z|^2,
\label{eq:M'6}  \\
M'_7 &=& abcdxy\Phi(a,b,c,d,z),\label{eq:M'7}
\end{eqnarray}
where
\begin{eqnarray} \label{eq:PolFi}
\Phi(a,b,c,d,z)=d^2(d^2-s_1)(d^2-s_1-|z|^2)-2abcd(x^2-y^2)-4s_3.
\end{eqnarray}

There is an $\SU(6)$-invariant, $J$, of bidegree $(1,3)$ whose
restriction to $W_6$ is given by
\begin{eqnarray} \label{eq:InvJ}
J'=2abcz+dz^*(M'_1-2d^2).
\end{eqnarray}
However, in order to uniquely identify $J$, we have to replace $W_6$
with the $7$-dimensional real subspace $W_7\supset W_6$, for which we
allow the coordinate $d$ in Eq. \eqref{eq:5dimVektor} to become
complex.  Thus, we replace $d$ with the complex coordinate $w$. By
denoting the restriction of $J$ to $W_7$ by $J''$, and similarly for
$M_1$, we have
\begin{eqnarray} \label{eq:InvJ''}
J''=2abcz+w^*z^*(M''_1-2|w|^2).
\end{eqnarray}

One can verify that
\begin{eqnarray} \label{eq:InvM7}
M_7= -\frac{1}{8}\Im(FJ^2).
\end{eqnarray}
We shall indicate later in this section how one can construct the
invariant $M_7$ without using $J$.  As $F$ is an $\SU(6)$-invariant of
bidegree $(4,0)$, it follows that $FJ^2$ is a unitary
invariant. Consequently, $M_7$ is also a unitary invariant. Moreover,
direct computation shows that
\begin{eqnarray} \label{eq:InvJ2}
|J|^2=\frac{1}{3}M_1(M_1M_2-2M_3)+M_2^2-4M_4-M_5.
\end{eqnarray}

\begin{theorem}\label{thm:Ferm-Generators}
The algebra $\cA:=\cP_\bR^\Un$ is generated over $\bC$ by the
invariants $M_1,\ldots,M_7$. The first six of these generators are
algebraically independent. If $\cA_p$ is the subalgebra of $\cA$
generated by these six polynomials, then $\cA$ is a free
$\cA_p$-module with basis $\{1,M_7\}$.
\end{theorem}
\begin{proof}
We have seen that the $M_i$ are unitary invariants. The algebraic
independence of the polynomials $M_1,\ldots,M_6$ can be verified by
exhibiting a point in $\we^3(V)$ at which the Jacobian of these
polynomials has rank $6$. In fact it suffices to verify this for the
restrictions $M'_1,\ldots,M'_6$. In that case one can use the point in
$W_6$ with coordinates $a=1$, $b=2$, $c=3$, $d=4$, and $x=y=1$. Hence,
$\cA_p$ is isomorphic to the polynomial algebra over $\bC$ in six
variables.  The homogeneous component, say $\cA_p^{(12)}$, of degree
$12$ of $\cA_p$ has dimension $12$, while the one of $\cA$ has
dimension $13$, see Eq. \eqref{eq:serUn}.  One can easily verify that
$M_7\notin\cA_p^{(12)}$.  If $f\in\cA_p\cap\cA_p M_7$, then $f=gM_7$
for some $g\in\cA_p$.  As $\cA_p$ is a unique factorization domain, it
is integrally closed in its field of fractions, say $K$. If $f\ne0$
then also $g\ne0$, and we obtain that $M_7=f/g\in K$. For convenience,
let $h=FJ^2$ and we point out that $|h|^2=hh^*=|F|^2|J|^4=M_5|J|^4$.
Eq. \eqref{eq:InvJ2} shows that $|J|^2\in\cA_p$, and so
$|h|^2\in\cA_p$. By a computer calculation, one can easily check that
$h+h^*$ belongs to $\cA_p$. Explicitly, we have
\begin{eqnarray} \label{eq:Reh}
9(h+h^*) &=& (M_1M_2 -2M_3)^2 +18M_2(M_2^2-4M_4-M_5)
+9M_1^2M_5 +144M_6.
\end{eqnarray}
This, and the identity $(h-h^*)^2=(h+h^*)^2-4|h|^2$, imply that
$M_7^2\in\cA_p$, and so $M_7$ is integral over $\cA_p$.  As $\cA_p$ is
integrally closed, we must have $M_7\in\cA_p$, which gives a
contradiction.

We conclude that $f=0$, and so $\cA_p\cap\cA_p M_7=0$.  Since the
Poincar\'{e} series of $\cA$ is the product of
$1+t^{12}$ and the Poincar\'{e} series of $\cA_p$, we must
have $\cA=\cA_p\oplus\cA_p M_7$.  In particular, the
invariants $M_1,\ldots,M_7$ generate $\cA$.
\end{proof}

Since the unitary invariants separate the LU-orbits \cite[Theorem 3,
  p. 133]{OV1990}, we have the following simple test for
LU-equivalence of pure fermionic states.
\begin{corollary}\label{cr:Equiv}
Two pure states $\ket{\ph},\ket{\ps}\in\we^3(V)$ are LU-equivalent if
and only if $M_i(\ph)=M_i(\ps)$ for $i=1,\ldots,7$.
\end{corollary}

From Eqs. \eqref{eq:M'5}, \eqref{eq:InvJ2}, and \eqref{eq:Reh}
we obtain that
\begin{eqnarray}  \notag
12^4M_7^2 &=& 36M_5 \left( M_1(M_1M_2-2M_3) +3M_2^2 -12M_4 -3M_5
 \right)^2 \\  \label{eq:syz}
&-& \left( (M_1M_2 -2M_3)^2 +18M_2(M_2^2-4M_4-M_5)
+9M_1^2M_5 +144M_6 \right)^2.
\end{eqnarray}
This is the unique algebraic relation (syzygy) among the seven
generators $M_i$, all other such relations are consequences of this
one. Thus we can construct the invariant $M_7$ (up to $\pm$ sign) from
this formula, and so avoid the use of the invariant $J$. The sign of
$M_7$ should be chosen to agree with Eq. \eqref{eq:M'7}.

We remark that the invariants $M_1,\ldots,M_7$ take real values
only. This follows from the fact that the restrictions $M'_i$ of the
$M_i$ are real-valued, see Eqs. \eqref{eq:M'1}--\eqref{eq:M'7}. While
$M_7$ obviously may take both positive and negative values, we claim
that $M_i\ge0$ for $i\ne7$. The claim is evidently valid for $i=1$ and
$i=5$, and for $i=2,4,6$ it follows from
Eqs. \eqref{eq:D-2}--\eqref{eq:D-6}. It remains to prove the claim for
$i=3$. It suffices to verify that the formula Eq. \eqref{eq:M'3} can
be written as
\begin{eqnarray}
M'_3 &=& x^2\left( (ab-cd)^2+(ac-bd)^2+(ad-bc)^2 \right)+
\notag \\
     & & y^2\left( (ab+cd)^2+(ac+bd)^2+(ad+bc)^2 \right)+S,
\end{eqnarray}
where $2S=\sum (\a^2-\b^2)^2\g^2$. The summation is over the $12$
ordered pairs $(\{\a,\b\},\g)$ where $\a,\b,\g$ are three distinct
elements of the set $\{a,b,c,d\}$. We omit the details of this
verification.

Let us also remark that $M_1M_2-2M_3\ge0$. Indeed, we have
\begin{eqnarray} \label{eq:M1M2-2M3}
\frac{1}{3}(M'_1M'_2-2M'_3)=(2abc-dy^2)^2+d^2|z|^2(M'_1-y^2)
+4s_2d^2+dx^2(4abc+dy^2),
\end{eqnarray}
which follows immediately from Eqs. \eqref{eq:M'1}--\eqref{eq:M'3}.

We recall that there are exactly four nonzero $\GL(6)$-orbits in
$\we^3(V)$. By analogy with the pure three-qubit states (see
\cite{LV08}) we label these orbits as follows and give their
representatives:

\noindent
\begin{tabular}{l@{\ }l}
(i)   & (fully) separable: $e_{246}$;\\
(ii)  & biseparable: $e_{235}+e_{246}$;\\
(iii) & W type: $e_{235}+e_{145}+e_{136}$;\\
(iv)  & GHZ type: $e_{135}+e_{246}$.
\end{tabular}

\noindent We warn the reader that, according to this definition, a
separable state is not biseparable.  We can determine the type of any
state by using the invariants.

\begin{proposition}\label{pp:type}
A (non-normalized) state $\ket{\psi}\in\we^3(V)$ is

\noindent
\begin{tabular}{l@{\ }l}
(i) & separable if and only if $M_2(\psi)=0$;\\
(ii) & biseparable if and only if $M_5$ and $M_1M_2-2M_3$ vanish at
  $\ket{\psi}$, and $M_2(\psi)>0$;\\
(iii) & of W type if and only if $M_5(\psi)=0$ and
$M_1(\psi)M_2(\psi)-2M_3(\psi)>0$;\\
(iv) & of GHZ type if and only if $M_5(\psi)>0$.
\end{tabular}
\end{proposition}
\begin{proof}
Without any loss of generality we may assume that $\ket{\psi}$ is
normalized and given by Eq. \eqref{eq:5dimVektor} where $a\ge b\ge
c\ge 0$, $d>0$ and $z=x+iy$. Note that $M_5(\psi)=0$ if and only if
$x=0$ and $4abc=dy^2$.
\begin{description}
\item[(i)] If $\ket{\psi}$ is separable, we may assume that $a=b=c=z=0$ and
$d=1$. Thus $M_2(\psi)=0$. Conversely, if $M_2(\psi)=0$ then
Eq. \eqref{eq:M'2} implies that $a=b=c=z=0$ and so $\ket{\psi}$ is
separable.
\item[(ii)] If $\ket{\psi}$ is biseparable, we may assume that $b=c=z=0$, and
so $M_5$ and $M_1M_2-2M_3$ vanish at $\ket{\psi}$.  Conversely, assume
that $M_5$ and $M_1M_2-2M_3$ vanish at $\ket{\psi}$ and that
$M_2(\psi)>0$. Then we have $x=0$ and Eq. \eqref{eq:M1M2-2M3} implies
that $b=c=z=0$. Now the condition $M_2(\psi)>0$ implies that $ad>0$,
and so $\ket{\psi}$ is biseparable.
\item[(iv)] It is well known that $\ket{\psi}$ is of GHZ type if and only if $F(\psi)\ne0$, i.e., $M_5(\psi)>0$.
\item[(iii)] This follows from (iv) and (ii).
\end{description}
\end{proof}
Note that for the decomposability of $\ket{\psi}$ it suffices to check
that $M_2(\psi)=0$ instead of using the Grassmann-Pl\"{u}cker
relations.

\section{Symmetric polynomial invariants of three qubits}
\label{sec:Inv3Qub}

The action of the local unitary group $\Un(2)\times\Un(2)\times\Un(2)$
on the Hilbert space of three qubits,
$\cH=\bC^2\mathrel{\otimes}\bC^2\mathrel{\otimes}\bC^2$, can be
extended to the action of the semidirect product
$G:=(\Un(2)\times\Un(2)\times\Un(2))\rtimes S_3$, where the symmetric
group $S_3$ permutes the three copies of $\bC^2$.  We refer to the
complex-valued $G$-invariant polynomial functions on $\cH$ (viewed as
a real vector space) as the {\em symmetric polynomial invariants} of
three qubits.  We denote by $\cB$ the $\bC$-algebra of the symmetric
polynomial invariants.

Our first objective in this section is to compute a minimal set of
generators of $\cB$. The complex linear map $\cH\to\we^3(V)$ which
maps
\begin{eqnarray} \label{eq:embedding}
\ket{ijk}\mapsto e_{i+1,j+3,k+5},\quad i,j,k\in\{0,1\}
\end{eqnarray}
is an isometry, i.e., it preserves the norms of vectors, as well as
the inner products. By using this map we shall identify $\cH$ with the
SOV subspace $W$.  We shall also identify $G$ with the subgroup of
$\Un(6)$ which preserves $W$. Our second objective is to prove that
the above embedding $\cH\to\we^3(V)$ establishes a one-to-one
correspondence between the $G$-orbits in $W$ and the $\Un(6)$-orbits
in $\we^3(V)$.

Let us recall some well-known facts about the algebra $\cC$ of
polynomial LU-invariants of three qubits. Let $\ket{\psi}\in\cH$ and
set $\r=\proj{\psi}$. There are six primary invariants:
\begin{eqnarray}
Q_1 &=& \tr\r,  \label{eq:Q1} \\
Q_2 &=& \tr(\r_A^2), \label{eq:Q2} \\
Q_3 &=& \tr(\r_B^2), \label{eq:Q3} \\
Q_4 &=& \tr(\r_C^2), \label{eq:Q4} \\
f_5 &=& \tr(\r_A\ox\r_B\r_{AB}), \label{eq:f5}  \\
Q_6 &=& |{\rm Hdet}|^2, \label{eq:Q6}
\end{eqnarray}
where Hdet is the Cayley hyperdeterminant, see e.g.,
\cite{ajt01}. There is only one secondary invariant
$f_7:=s_2^2{\rm Hdet}^*$. This invariant has been constructed
by M. Grassl, see \cite{ajt01} for the definition of $s_2$. For
convenience, we shall replace the generators $f_5$ and $f_7$ by the
invariants
\begin{eqnarray}
Q_5 &=& \tr(\r_A\ox\r_B\r_{AB})+\tr(\r_B\ox\r_C\r_{BC})
 +\tr(\r_A\ox\r_C\r_{AC}), \label{eq:Q5} \\
Q_7 &=& \Im(f_7). \label{eq:Q7}
\end{eqnarray}

The values of the $Q_i$ at $\ket{\psi}=a\ket{100}+b\ket{010}+
c\ket{001}+d\ket{111}+z\ket{000}$, where $a,b,c,d$ are real and
$z=x+iy$, are given by the formulae:
\begin{eqnarray}
Q_1 &=& s_1+d^2+|z|^2,  \label{eq:Q1'} \\
Q_2 &=& (a^2+d^2)^2+(b^2+c^2)^2+2s_1|z|^2+|z|^4,
\label{eq:Q2'} \\
Q_3 &=& (b^2+d^2)^2+(a^2+c^2)^2+2s_1|z|^2+|z|^4,
\label{eq:Q3'} \\
Q_4 &=& (c^2+d^2)^2+(a^2+b^2)^2+2s_1|z|^2+|z|^4,
\label{eq:Q4'} \\
Q_5 &=& 3|z|^6+9s_1|z|^4+ \left( 9(a^4+b^4+c^4) +11s_2 +2s_1d^2
\right) |z|^2 +6abcd(x^2-y^2)+ \notag \\
&& 3(a^6+b^6+c^6+d^6) +2s_1d^4 +2(a^4+b^4+c^4)d^2 +3s_2d^2
+2s_1s_2 -3s_3, \label{eq:Q5'}  \\
Q_6 &=& d^2 \left( (4abc-d|z|^2)^2 +16abcdx^2 \right),
\label{eq:Q6'} \\
Q_7 &=& 8abcdxy\Phi(a,b,c,d,z), \label{eq:Q7'}
\end{eqnarray}
where $\Phi$ is given in Eq. \eqref{eq:PolFi}.  It is still true that
$Q_1,\ldots,Q_6$ can be taken as the primary invariants and $Q_7$ as
the secondary invariant.  Since the transposition $(1,2)\in S_3$ fixes
$f_5$, it follows that $Q_5$, which is the symmetrization of $f_5$
with respect to the cycle $(1,2,3)\in S_3$, is fixed by $S_3$.  While
$f_7$ is not fixed by $S_3$, one can check that $Q_7$ is. Moreover, a
computation shows that $Q_7^2\in\bC[Q_1,\ldots,Q_6]$. Note that now
$S_3$ fixes the LU-invariants $Q_1,Q_5,Q_6,Q_7$, and one can check
that it acts faithfully on the set $\{Q_2,Q_3,Q_4\}$. Thus, $S_3$ acts
on the polynomial algebra $\bC[Q_2,Q_3,Q_4]$. Since any symmetric
polynomial can be expressed as a polynomial in the elementary
symmetric functions of the variables, we conclude that the
$S_3$-invariants in $\bC[Q_2,Q_3,Q_4]$ are generated by
 \begin{eqnarray}
Q_2+Q_3+Q_4,\quad  Q_2Q_3+Q_2Q_4+Q_3Q_4,\quad Q_2Q_3Q_4.
\end{eqnarray}

We can now construct a minimal set of generators for the algebra
of polynomial $G$-invariants of three qubits.
\begin{theorem}\label{thm:algebraB}
The algebra $\cB$ of polynomial $G$-invariants of three qubits is
generated by the polynomials $Q_1$, $Q_2+Q_3+Q_4$, $Q_5$,
$Q_2Q_3+Q_2Q_4+Q_3Q_4$, $Q_6$, $Q_2Q_3Q_4$, and $Q_7$ of degree $2$,
$4$, $6$, $8$, $8$, $12$, and $12$, respectively. The first six are
the primary invariants and they generate the subalgebra $\cB_p$
isomorphic to the polynomial algebra in six variables. The last
generator, $Q_7$, is a secondary invariant. Moreover, $Q_7^2\in\cB_p$
and $\cB$ is a free module over $\cB_p$ with basis $\{1,Q_7\}$.
\end{theorem}
\begin{proof}
Let $\cC_p$ be the subalgebra of $\cC$ generated by its primary
generators $Q_1,\ldots,Q_6$. Note that $\cC_p$ is invariant under the
action of $S_3$. Since $\cC_p=\bC[Q_1,Q_5,Q_6]\ox\bC[Q_2,Q_3,Q_4]$ and
$S_3$ fixes $Q_1$, $Q_5$, and $Q_6$, we deduce that the subalgebra,
$\cC_p^{S_3}$, of $S_3$-invariants in $\cC_p$ is generated by
$Q_1,Q_5,Q_6$, $Q_2+Q_3+Q_4$, $Q_2Q_3+Q_2Q_4+Q_3Q_4$, and
$Q_2Q_3Q_4$. Thus we have $\cC_p^{S_3}=\cB_p$. Since these six
polynomials are algebraically independent, $\cB_p$ is isomorphic to
the polynomial algebra in six variables.

Since $Q_7$ is fixed by $S_3$ and $Q_7^2\in\cC_p$, we deduce that
$Q_7^2\in\cC_p^{S_3}=\cB_p$. Since $\cC$ is a free $\cC_p$ module with
basis $\{1,Q_7\}$, we can write any $f\in\cB$ uniquely as
$f=f_1+f_2Q_7$ with $f_1,f_2\in\cC_p$. For any $\s\in S_3$ we have
$f=f_1^\s+f_2^\s Q_7$, which implies that $f_1^\s=f_1$ and
$f_2^\s=f_2$. Consequently, $f_1,f_2\in\cB_p$. We conclude that
$f\in\cB_p[Q_7]$, and so $\cB=\cB_p[Q_7]$. As in the proof of Theorem
\ref{thm:Ferm-Generators}, we can show that $\cB$ is a free $\cB_p$
module with basis $\{1,Q_7\}$.
\end{proof}
We can now show that the embedding \eqref{eq:embedding} gives a
one-to-one correspondence between $G$-equivalence classes of pure
three-qubit states in $\cH$ and the LU-equivalence classes of pure
fermionic states in $\we^3(V)$.
\begin{theorem}\label{th:equiv}
The restriction map $f\to f|_W$ from the algebra $\cA:=\cP_\bR^\Un$ to
the $\bC$-algebra of polynomial functions on $W$, viewed as a real
vector space, is injective and its image is the algebra $\cB$ of
polynomial $G$-invariants of three qubits. In particular these two
algebras are isomorphic as graded algebras. If $\cO$ is an
$\Un(6)$-orbit in $\we^3(V)$ then $\cO\cap W$ is a single $G$-orbit in
$W$, and the map which sends $\cO$ to $\cO\cap W$ is a one-to-one
correspondence between the LU-equivalence classes of pure fermionic
states in $\we^3(V)$ and the $G$-equivalence classes of pure $3$-qubit
states in $W$.
\end{theorem}
\begin{proof}
For $f\in\cA$ we denote by $f_W$ the restriction of $f$ to the
subspace $W\subseteq\we^3(V)$. Let us also denote by $\cA_W$ the image
of $\cA$ by this restriction map. Recall that each pure fermionic
state is LU-equivalent to an SOV. This implies that if $f_W=0$ and
$f\in\cA$, then $f=0$. Thus, the restriction map $\cA\to\cA_W$ is an
isomorphism of graded algebras. As $G$ is a subgroup of $\Un(6)$, it
is easy to verify that for $f\in\cA$ we have $f_W\in\cB$. Hence,
$\cA_W$ is a subalgebra of $\cB$.  It follows from Theorem
\ref{thm:algebraB} that the algebras $\cA_W$ and $\cB$ have the same
Poincar\'{e} series, namely the one given by Eq. \eqref{eq:serUn}. As
$\cA_W\subseteq\cB$, we must have the equality $\cA_W=\cB$.

Now let $\cO$ be an $\Un(6)$-orbit in $\we^3(V)$. As mentioned above,
we must have $\cO\cap W\ne\es$. Since $\cA_W=\cB$, it follows that
$\cO\cap W$ is a single $G$-orbit. Hence, the map which sends
$\cO\to\cO\cap W$ is indeed the one-to-one correspondence as asserted
in the theorem.
\end{proof}
\begin{corollary}\label{cr:Equivalence}
A polynomial $f\in\cC$ is invariant under qubit permutations (i.e., it
belongs to the subalgebra $\cB$) if and only if its restriction
$f':=f|_{W_6}$ is a symmetric polynomial with respect to the variables
$a$, $b$, $c$.
\end{corollary}
\begin{proof}
Recall that $W_6$ is an $S_3$-invariant real subspace of $W$ of
dimension six, with coordinates $a$, $b$, $c$, $d$, $x$, $y$. It is
easy to verify that, when restricted to $W_6$, $S_3$ fixes the
coordinates $d$, $x$, $y$ and permutes the coordinates $a$, $b$, $c$.
Now the assertion follows from the fact that the map which sends
$f\in\cC$ to its restriction $f'$ is injective.
\end{proof}

The one-to-one correspondence mentioned in Theorem \ref{th:equiv}
is in fact a homeomorphism between the orbit spaces
$\bwe^3(V)/\Un(6)$ and $W/G$. As sets, these orbit spaces are
just the set of $\Un(6)$-orbits in $\bwe^3(V)$ and the set of
$G$-orbits in $W$, respectively. However, they are also toplogical spaces with respect to the quotient topology arising
from the projection maps $\pi_V:\bwe^3(V)\to\bwe^3(V)/\Un(6)$ and
$\pi_W:W\to W/G$. For additional properties of orbit spaces of
compact Lie groups we refer the reader to
\cite[Chapter I]{gb:1972}.

\begin{corollary}\label{cr:OrbSpaces}
The one-to-one correspondence $\bwe^3(V)/\Un(6)\to W/G$,
constructed in Theorem \ref{th:equiv} is a homeomorphism.
\end{corollary}
\begin{proof}
Let us denote by $T$ the inverse of this correspondence.
Thus, for any $\ket{\phi}\in W$ we have
$T(G\cdot\ket{\phi})=\Un(6)\cdot\ket{\phi}$. Since the
composite of $\pi_V$ and the inclusion map $\iota:W\to\bwe^3(V)$ is continuous and coincides with $T\circ\pi_W$, we deduce that
$T$ is continuous. It remains to show that $T^{-1}$ is continuous. By Lemma \ref{lm:5basis} we know that
$\Un(6)\cdot W=\bwe^3(V)$. Let $\Un(6)$ act trivially on $W/G$. If $\ket{\a},\ket{\b}\in W$ are LU-equivalent, then by Theorem
\ref{th:equiv} they belong to the same $G$-orbit and so
$\pi_W(\a)=\pi_W(\b)$. Therefore we can apply
\cite[Theorem 3.3]{gb:1972} to the $\Un(6)$-spaces
$\bwe^3(V)/\Un(6)$ and $W/G$. We conclude that
$\pi_W$ extends uniquely to a continuous map
$\pi'_W:\bwe^3(V)\to W/G$ which is constant on $\Un(6)$-orbits.
Since $\pi_V\circ\iota=T\circ\pi_W=T\circ\pi'_W\circ\iota$,
it follows easily that $\pi_V=T\circ\pi'_W$. Hence,
$\pi'_W=T^{-1}\circ\pi_V$ and so $T^{-1}$ must be continuos.
\end{proof}

\section{The canonical region}
\label{sec:KanReg}

In this section we shall work with normalized states only, which means
that we have $M_1=1$. For convenience, we shall identify any point
$(a,b,c,d,z)\in\bR^4\times\bC$ with the corresponding vector
$\ket{\psi}$ given by Eq. \eqref{eq:5dimVektor}. We shall always write
the complex coordinate $z$ as $z=x+iy$, where $x$ and $y$ are
real. Thus we have identified the subspace $W_6$ with the product
$\bR^4\times\bC$.  We denote by $\S$ the unit sphere of $W_6$, i.e.,
the set of all points $(a,b,c,d,z)\in W_6$ such that
$a^2+b^2+c^2+d^2+|z|^2=1$.  We shall denote by $\S_i$, $i=1,2,3$, the
unit sphere of the subspace $V_i$ of $V$, see Eq. \eqref{eq:Vi}.

\subsection{The canonical region $\Delta$}

To begin, we introduce two continuous functions $\mu\colon
\we^3(V)\to\bR$ and $\mu'\colon W\to\bR$.
\begin{definition} \label{def:mu}
For $\ket{\psi}\in\we^3(V)$, we set
\begin{eqnarray} \label{eq:mu}
\mu(\ps) &=& \max_{\a,\b,\g} |{\braket{\a\we\b\we\g}{\ps}}|, \quad \| \a\we\b\we\g \|=1,
\end{eqnarray}
and for $\ket{\psi}\in W$, we set
\begin{eqnarray} \label{eq:mu'}
\mu'(\ps) &=& \max_{\a,\b,\g} |{\braket{\a\we\b\we\g}{\ps}}|, \quad (\a,\b,\g)\in\S_1\times\S_2\times\S_3.
\end{eqnarray}
\end{definition}

We prove that $\mu'$ is the restriction of $\mu$.
\begin{lemma}\label{le:maximum}
For $\ket{\psi}\in W$, we have $\mu'(\ps)=\mu(\ps)$.
\end{lemma}
\begin{proof}
Note that $\mu'(\ps)\le\mu(\ps)$.  Let $\ket{\a\we\b\we\g}$ be a unit
vector such that $\mu(\ps)=|\braket{\a\we\b\we\g}{\ps}|$.  Let
$L=\lin\{\a,\b,\g\}$ and choose unit vectors
$\ket{\a'},\ket{\b'},\ket{\g'}\in L$ such that $\ket{\a'}\in V_2+V_3$,
$\ket{\b'}\in V_1+V_3$, and $\ket{\g'}\in V_1+V_2$. At least two of
these three vectors must be linearly independent, say $\ket{\a'}$ and
$\ket{\g'}$.  We can write $\ket{\a'}=c_2\ket{\g_2}+c_3\ket{\g_3}$
where $c_2,c_3\ge0$, and $\ket{\g_2}\in V_2$ and $\ket{\g_3}\in V_3$
are unit vectors. Similarly, we write
$\ket{\g'}=a_1\ket{\a_1}+a_2\ket{\a_2}$ where $a_1,a_2\ge0$, and
$\ket{\a_1}\in V_1$ and $\ket{\a_2}\in V_2$ are unit vectors.  As
$\mu(\ps)>0$, either $a_1$ or $c_3$ is non-zero, and so we have
$a_2c_2<1$. We have $|\braket{\a'}{\g'}|^2=ta_2^2c_2^2$ where
$t:=|\braket{\a_2}{\g_2}|^2\le1$. Finally, let
$\ket{\b''}=b_1\ket{\b_1}+b_2\ket{\b_2}+b_3\ket{\b_3}\in L$, with
$\ket{\b_i}\in V_i$ and $b_i\ge0$, be a nonzero vector orthogonal to
$\ket{\a'}$ and $\ket{\g'}$. We fix its norm to be
$\|\b''\|=1/\sqrt{1-ta_2^2c_2^2}$.  Since both $\ket{\a\we\b\we\g}$
and $\ket{\g'\we\b''\we\a'}$ are unit vectors in $\we^3(L)$, they
differ only by a phase factor.  Hence, we can assume that
$\ket{\a}=\ket{\g'}$, $\ket{\b}=\ket{\b''}$ and $\ket{\g}=\ket{\a'}$.
Since
\begin{eqnarray}\label{ea:abg}
|\braket{\a\we\b\we\g}{\psi}|&\le&
a_1b_2c_3|\braket{\a_1\we\b_2\we\g_3}{\psi}|\nonumber\\
&&\quad+a_1b_3c_2|\braket{\a_1\we\b_3\we\g_2}{\psi}|
 +a_2b_1c_3|\braket{\a_2\we\b_1\we\g_3}{\psi}|,
\end{eqnarray}
we infer that $\mu(\psi)\le(a_1b_2c_3+a_1b_3c_2+a_2b_1c_3)
\mu'(\psi)$. We have $\mu(\psi)\le\mu'(\psi)$ because
\begin{eqnarray} \notag
 (a_1b_2c_3+a_1b_3c_2+a_2b_1c_3)^2
 &\le&
 (a_1^2c_3^2+a_1^2c_2^2+a_2^2c_3^2)
 (b_1^2+b_2^2+b_3^2) \\  \label{ea:abc=1}
  &=&
  \frac{1- a_2^2c_2^2}{1-ta_2^2c_2^2}  \le 1.
\end{eqnarray}
Consequently, $\mu'(\psi)=\mu(\psi)$.
\end{proof}

For $\ket{\psi}\in W_6$ given by Eq. \eqref{eq:5dimVektor}, we have
\begin{eqnarray} \label{eq:muW6}
\mu(\psi)^2 &=& \max_{u_1,u_2,u_3} |\braket{\psi}{
(\x_1\ket{1}+\eta_1\ket{2})\we
(\x_2\ket{3}+\eta_2\ket{4})\we
(\x_3\ket{5}+\eta_3\ket{6})}|^2 \notag \\
&=& \max_{u_1,u_2,u_3}
|a\eta_1\x_2\x_3+b\x_1\eta_2\x_3+c\x_1\x_2\eta_3+
d\eta_1\eta_2\eta_3+z^*\x_1\x_2\x_3|^2,
\end{eqnarray}
where the maximum is taken over all unit vectors
$u_i=(\x_i,\eta_i)\in\bC^2$, $i=1,2,3$.

\begin{definition}\label{def:Delta}
We denote by $\D$ the subset of $\S$ consisting of all points
$p=(a,b,c,d,z)$, $z=x+iy$, such that (i) $a\ge b\ge c\ge0$, $x\ge0$
and (ii) $d=\mu(\psi)$, where the state $\ket{\psi}$ corresponds to
$p$.
\end{definition}
Since $\mu$ is a continuous function, it follows that $\D$ is a closed
(and so compact) subset of $\S$. The relative interior of $\Delta$ as
a subset of $\S$, is the set $\Delta^0$ of all points $p\in\Delta$ for
which there exists $\varepsilon>0$ such that $q\in\S$ and
$\|q-p\|<\varepsilon$ imply that $q\in\D$. As $\D$ is closed, the
relative boundary $\partial\D$ of $\D$ is the set-theoretic difference
$\D\setminus\D^0$. At this point we do not know whether
$\D^0\ne\emptyset$, but we shall see later that this is the case.

Note that $\D$ is contained in the northern hemisphere of $\S$ defined
by the inequality $d>0$.  We shall prove now that $\D$ and $\D^0$ are
connected.  For convenience, we shall identify in this lemma the points
$p\in\bR^4\times\bC$ with the corresponding vectors $\vec{0p}$.
\begin{lemma} \label{le:connection}
Let $p,q\in\D$ and let $r=tp+(1-t)q$ where $0<t<1$. Then the point
$s=r/\|r\|$ belongs to $\D$. Consequently, $\D$ and $\D^0$ are
connected.
\end{lemma}
\begin{proof}
Since $p$ and $q$ satisfy the linear inequalities (i) of Definition
\ref{def:Delta}, so do $r$ and $s$. It remains to prove that $s$ also
satisfies the condition (ii) of that definition. Let $\ket{\ps_1}$,
$\ket{\ps_2}$, and $\ket{\ps_3}$ be the states corresponding to the
points $p$, $q$, and $s$, respectively, and note that
$\|r\|\cdot\ket{\ps_3}=t\ket{\ps_1}+(1-t)\ket{\ps_2}$. Let
$p=(a_1,b_1,c_1,d_1,z_1)$, $q=(a_2,b_2,c_2,d_2,z_2)$, and
$s:=(a_3,b_3,c_3,d_3,z_3)$. By definition of $\mu$ we have
\begin{eqnarray}
\mu(s) &=& \max_{\a,\b,\g} |{\braket{\a\we\b\we\g}{\ps_3}}|
 \notag\\
 &\le& \frac{1}{\norm{r} } \bigg( t\cdot\max_{\a,\b,\g}
 |\braket{\a\we\b\we\g}{\ps_1}|
 +
(1-t)\cdot\max_{\a,\b,\g}
 |\braket{\a\we\b\we\g}{\ps_2}| \bigg)
 \notag\\
 &=&\frac{1}{\norm{r} } (t d_1 + (1-t) d_2) = d_3,
\end{eqnarray}
where the maxima are over all decomposable $3$-vectors
$\ket{\a\we\b\we\g}$ of unit norm. As the inequality $\mu(s)\ge d_3$
is trivial, the proof is completed.
\end{proof}

In the following lemma we prove a basic property of $\D$.
\begin{lemma}\label{le:KanForm}
Each LU-orbit of normalized fermionic states
$\ket{\psi}\in\wedge^3(V)$ has a representative in $\Delta$.
\end{lemma}
\begin{proof}
Let $\ket{\psi}\in\we^3(V)$ be a unit vector. By Lemma \ref{lm:5basis}
the LU-orbit of $\ket{\psi}$ meets $W$, and so we may assume that
$\ket{\psi}\in W$, where $W=V_1\we V_2\we V_3$ is the SOV subspace
(see Eq. \eqref{eq:Vi}). By Lemma \ref{le:maximum} there exists unit
vectors $\ket{\a}\in V_1$, $\ket{\b}\in V_2$, and $\ket{\g}\in V_3$
such that $\mu(\ps)=\braket{\a\we\b\we\g}{\ps}$. Let
$g\in\Un(2)\times\Un(2)\times\Un(2)$ be chosen so that
$g\ket{\a}=\ket{2}$, $g\ket{\b}=\ket{4}$, and $g\ket{\g}=\ket{6}$. By
replacing $\ket{\psi}$ with $g\cdot\ket{\psi}$, we obtain that
$\mu:=\mu(\psi)=\braket{e_{246}}{\psi}$. Thus we can write
$\ket{\psi}=\sum_{i,j,k\in\{0,1\}}t_{ijk}e_{ijk}$, where $t_{ijk}\in\bC$,
with $t_{111}=\mu$.

It is now easy to show that $t_{ijk}=0$ if exactly one of $i,j,k$ is
zero. Assume that $i=0$ and $j=k=1$. Let
$\ket{\ph}=\ket{\a\we\b\we\g}$ where
$\a=(t_{011}\ket{1}+\mu\ket{2})/\sqrt{|t_{011}|^2+\mu^2}$,
$\b=\ket{4}$, and $\g=\ket{6}$. Then we have
$\mu\ge|\braket{\ph}{\ps}|=\sqrt{|t_{011}|^2+\mu^2}$, and so
$t_{011}=0$. Consequently $\ket{\psi}\in W_6$.

By applying a suitable diagonal LU transformation, we can further
assume that $t_{ijk}\ge0$ if exactly two of $i,j,k$ are zero. By
using the action of the symmetric group $S_3$ described in
Sec. \ref{sec:RDM}, we can further assume that $t_{100}\ge t_{010}\ge
t_{001}\ge0$. Finally, the diagonal LU transformation, which fixes the
basis vectors $\ket{i}$ for $i=2,4,6$ and multiplies them with $-1$
for $i=1,3,5$, will replace $t_{000}$ with $-t_{000}$ and will not
change any other $t_{ijk}$.  This means that we can assume that also
$\Re(t_{000})\ge0$. Hence $\ket{\psi}\in\Delta$.
\end{proof}

In view of this lemma, we shall refer to $\D$ as the {\em canonical
  region}. Evidently, two LU-equivalent states in $\D$ have the same
coefficient $d$. Thus
\begin{corollary}\label{cr:delta=DIFFERENTd}
Two points in $\D$ having different $d$-coordinate are not LU-equivalent.
\end{corollary}

\subsection{Inequalities defining $\D$}

Our definition of $\D$ is not easy to use because the condition (ii)
is hard to verify. We shall prove that this condition can be replaced
by the inequality $d>0$ and four additional inequalities. The set of
inequalities defining $\D$ as a subset of $\S$ will be simplified
later, see Proposition \ref{pp:projection}. We need an auxiliary
lemma, but first let us make a couple of observations.

If $\a,\b,\g,\d$ are arbitrary complex numbers, then the following identity holds
\begin{eqnarray} \label{eq:C-identity}
(|\a|^2+|\b|^2+|\g|^2+|\d|^2)^2-4|\a\d-\b\g|^2 =
(|\a|^2-|\b|^2+|\g|^2-|\d|^2)^2+4|\a\b^*+\g\d^*|^2.
\end{eqnarray}

By the Cauchy-Schwarz inequality we have
\begin{eqnarray} \label{eq:C-S}
\max_{u=(\x,\eta)} |\a\x+\b\eta|^2 = |\a|^2+|\b|^2,
\end{eqnarray}
where $\a$ and $\b$ are arbitrary complex numbers, and the maximum is
taken over all unit vectors $u=(\x,\eta)\in\bC^2$. This observation
can be generalized as follows.
\begin{lemma}\label{le:C-S}
Let $\a,\b,\g,\d$ be any complex numbers. Denote the maximum of
$|(\a\x_1+\b\eta_1)\x_2+(\g\x_1+\d\eta_1)\eta_2|^2$ taken over all
unit vectors $u_i=(\x_i,\eta_i)\in\bC^2$, $i=1,2$, by $\lambda$. Then
\begin{eqnarray} \label{eq:C-Smax}
\lambda=\frac{1}{2} \left( |\a|^2+|\b|^2+|\g|^2+|\d|^2+
\sqrt{(|\a|^2+|\b|^2+|\g|^2+|\d|^2)^2-4|\a\d-\b\g|^2}
\right).
\end{eqnarray}
\end{lemma}
\begin{proof}
By Eq. \eqref{eq:C-S} we have
\begin{eqnarray}
\lambda=\max_{u_1} \left( |\a\x_1+\b\eta_1|^2 +
|\g\x_1+\d\eta_1|^2 \right).
\end{eqnarray}
The function that we are maximizing here can be rewritten as
\begin{eqnarray}
&& \frac{1}{2}\left(|\a|^2+|\b|^2+|\g|^2+|\d|^2\right)
+\frac{1}{2}\left(|\a|^2-|\b|^2+|\g|^2-|\d|^2\right)
\left(|\x_1|^2-|\eta_1|^2\right) \notag \\
&&\qquad +2\Re \left( (\a\b^*+\g\d^*)\x_1\eta_1^* \right).
\end{eqnarray}
Since $(|\x_1|^2-|\eta_1|^2, 2\x_1\eta_1^*)$ runs through
all unit vectors in $\bC^2$ (up to an overall phase factor),
another application of Eq. \eqref{eq:C-S} gives
\begin{eqnarray} \label{eq:C-Smaxi}
\lambda=\frac{1}{2} \left( |\a|^2+|\b|^2+|\g|^2+|\d|^2+
\sqrt{(|\a|^2-|\b|^2+|\g|^2-|\d|^2)^2+4|\a\b^*+\g\d^*|^2}
\right).
\end{eqnarray}
By using the identity \eqref{eq:C-identity}, this formula can be rewritten in the form \eqref{eq:C-Smax}.
\end{proof}

Note that $a^2+b^2+c^2+d^2+|z|^2=1$ implies the equality
\begin{eqnarray} \label{eq:alpha-id}
(d^2-a^2)(d^2-b^2) -d^2|z|^2=a^2b^2+c^2d^2 +d^2(2d^2-1),
\end{eqnarray}
which will be tacitly used in the next proof.

\begin{theorem}\label{thm:NejDelta}
Let $p=(a,b,c,d,z)\in\S$, $z=x+iy$, such that $a\ge b\ge c\ge0$,
$x\ge0$, and let $\ket{\psi}$ be the corresponding normalized pure
fermionic state.
\begin{description}
\item[(a)] If $p\in\D$, then $d>0$ and the following inequalities hold
\begin{eqnarray}
&& a^2b^2+c^2d^2 +d^2(2d^2-1)\ge0, \label{eq:Nej(i)} \\
&& d(d^2-s_1)-2abc\ge0, \label{eq:Nej(ii)} \\
&& 2\left( a^2b^2+c^2d^2 +d^2(2d^2-1) \right)(d^2-s_1)
   -x^2(ab+cd)^2 -y^2(ab-cd)^2\ge0, \label{eq:Nej(iii)} \\
&& d^2(2d^2-1)(d^2-s_1) -2abcd(x^2-y^2) -4s_3 \ge0.
\label{eq:Nej(iv)}
\end{eqnarray}
If equality holds in \eqref{eq:Nej(i)}, then $bx=(ab-cd)y=0$.
If equality holds in \eqref{eq:Nej(ii)} or \eqref{eq:Nej(iii)}, then $bx=0$.
\item[(b)] Conversely, if $d>0$ and the inequalities
\eqref{eq:Nej(i)}--\eqref{eq:Nej(iv)} hold, then $p\in\D$.
\end{description}
\end{theorem}
\begin{proof}
To simplify notation, we shall denote by $\a,\b,\g,\d$ the left hand
sides of \eqref{eq:Nej(i)}--\eqref{eq:Nej(iv)}, respectively.  By
applying Lemma \ref{le:C-S} to compute the maximum in
Eq. \eqref{eq:muW6} over unit vectors $u_1$ and $u_2$ only, we obtain
that
\begin{eqnarray}
 \label{eq:mu=max}
 \mu(\psi)^2=\frac{1}{2}\max_{u_3} \left( P+\sqrt{P^2 -4|Q|^2} \right),
\end{eqnarray}
where $P=|z^*\x_3+c\eta_3|^2 +(a^2+b^2)|\x_3|^2 +d^2|\eta_3|^2$,
$Q=ab\x_3^2-dz^*\x_3\eta_3-cd\eta_3^2$, and
$u_3=(\x_3,\eta_3)\in\bC^2$ is any unit vector.

Proof of (a). As $p\in\D$, we have $d=\mu(\psi)>0$. Consequently, the inequality
\begin{eqnarray} \label{ea:max(u3)}
 2d^2\ge P+\sqrt{P^2 -4|Q|^2}
\end{eqnarray}
is valid for all unit vectors $u_3=(\x_3,\eta_3)\in\bC^2$. We now
assume that $\eta_3\ne0$ and set $\x_3/\eta_3=te^{i\theta}$ where $t$
and $\theta$ are real numbers. By moving the term $P$ in
\eqref{ea:max(u3)} to the left hand side and squaring both sides, we
obtain that
\begin{eqnarray} \label{eq:Nej-3}
d^4-d^2 \left( |tz^*e^{i\theta}+c|^2 +(a^2+b^2)t^2 +d^2 \right)
|\eta_3|^2
+|abt^2e^{2i\theta}-dz^*te^{i\theta}-cd|^2\cdot|\eta_3|^4\ge0.
\end{eqnarray}
We divide the left hand side by $|\eta_3|^4$ and use the fact that
$|\eta_3|^{-2}=1+t^2$. Then the left hand side is divisible by $t^2$,
and by omitting this factor we obtain that
\begin{eqnarray} \label{nej:osn}
\a t^2 -2d \Re \left( (abz+cdz^*)e^{i\theta} \right) t
+d \left( d(d^2-s_1)-2abc\cos 2\theta \right) \ge0.
\end{eqnarray}
Since this inequality holds for all real $t$ and $\theta$, we deduce
that $\a\ge0$ and $\b\ge0$, i.e., the inequalities \eqref{eq:Nej(i)}
and \eqref{eq:Nej(ii)} hold. As the discriminant of the quadratic
polynomial in $t$ on the left hand side of \eqref{nej:osn} must be
nonpositive for all real $\theta$, we have
\begin{eqnarray} \label{eq:NejDiskr}
\a \left( d(d^2-s_1)-2abc \cos 2\theta \right)
-d\left( \Re\left((abz+cdz^*)e^{i\theta}\right) \right)^2 \ge0.
\end{eqnarray}
If $\a=0$, then we must have $abz+cdz^*=0$, which is equivalent to
$bx=(ab-cd)y=0$. Now assume that $\b=0$. By setting $\theta=0$ in
\eqref{eq:NejDiskr}, we deduce that $\Re(abz+cdz^*)=0$, i.e., $bx=0$.

Next we shall prove \eqref{eq:Nej(iii)} and \eqref{eq:Nej(iv)}. We can
rewrite the inequality \eqref{eq:NejDiskr} as
\begin{eqnarray} \label{eq:Nej-2}
d\g\ge \left( 4\a abc +dx^2(ab+cd)^2 -dy^2(ab-cd)^2 \right)
\cos 2\theta -2dxy(a^2b^2-c^2d^2) \sin 2\theta.
\end{eqnarray}
Since this inequality holds for all real $\theta$, it follows
that $\g\ge0$, i.e., the inequality \eqref{eq:Nej(iii)}
holds. Moreover, we must have
\begin{eqnarray}
d^2\g^2   \label{eq:Nej-1}
-\left( 4\a abc +dx^2(ab+cd)^2 -dy^2(ab-cd)^2 \right)^2
-\left( 2dxy(a^2b^2-c^2d^2) \right)^2 \ge0.
\end{eqnarray}

If $\g=0$, then the two equalities $xy(a^2b^2-c^2d^2)=0$ and $4\a abc
+dx^2(ab+cd)^2 -dy^2(ab-cd)^2 =0$ must hold. If also $x>0$ then from
the first equality we deduce that $y(ab-cd)=0$, and then from the
second one we deduce that $b=0$. Thus $\g=0$ implies that $bx=0$.

The left hand side of \eqref{eq:Nej-1} factorizes and we obtain that
\begin{eqnarray} \label{eq:fakt}
4(d^2-a^2)(d^2-b^2)\a\d \ge0.
\end{eqnarray}

Assume that $\d<0$. Then $d=a$ or $\a=0$. If $d=a$, then $\b\ge0$
implies that $b=c=0$ and $d^2-s_1=0$, which contradicts the
assumption that $\d<0$. Hence, we must have $d>a$ and $\a=0$. It
follows that $bx=(ab-cd)y=0$. If $b=0$, then
$\d=d^2(2d^2-1)(d^2-a^2)<0$ contradicts the inequality
\eqref{eq:Nej(i)}. Hence $b>0$, and so $x=0$ and
$(d^2-a^2)(d^2-b^2)=d^2y^2$. As $d>a$, this implies that $y\ne0$ and
so $ab=cd$ and $c>0$. Now \eqref{eq:Nej(i)}  implies that
$c^2+d^2=1/2$. The inequality $\d<0$ becomes
$-2c^2d^2(d^2-s_1)+2c^2d^2y^2-4c^4d^2<0$. After canceling the
factor $2c^2d^2$, we obtain that $a^2+b^2+y^2<1/2$ which gives a
contradiction. Hence, we must have $\d\ge0$.

Proof of (b). We have to prove that $d=\mu(\psi)$. It is immediate
from the definition of $\mu(\psi)$ that $d\le\mu(\psi)$. In order to
prove that $d\ge\mu(\psi)$, we will reverse the main steps in the
proof of (a).

Since $d>0$, it follows from $\b\ge0$ that $d^2\ge s_1$ and, in
particular, $d\ge a$. Since $\a\ge0$ and $\d\ge0$, the inequality
\eqref{eq:fakt} holds, and so does the inequality \eqref{eq:Nej-1}. As
$\g\ge0$, we deduce that the inequality \eqref{eq:Nej-2} holds for all
real $\theta$, and so does \eqref{eq:NejDiskr}.  Since $\a\ge0$ and
$\b\ge0$, this implies that the inequality \eqref{nej:osn} holds for
all real $t$ and $\theta$, and that the inequality \eqref{eq:Nej-3}
holds for all unit vectors $u_3=(\x_3,\eta_3)\in\bC^2$ with
$\x_3=te^{i\theta}\eta_3$.  We can rewrite the inequality
\eqref{eq:Nej-3} as $(2d^2-P)^2\ge P^2-4|Q|^2$. By using the identity
\eqref{eq:C-identity}, one can easily show that $P^2-4|Q|^2\ge0$.

We claim that $2d^2-P\ge0$ for all unit vectors $u_3$.  By computing
the maximum of $P$ over all $u_3$, our claim asserts that $2d^2\ge
1/2+((1/2-c^2-d^2)^2+c^2|z|^2)^{1/2}$.  Note that \eqref{eq:Nej(i)}
implies that $d^2\ge a^2+|z|^2$.  As $2d^2\ge b^2+c^2$, we obtain that
$4d^2\ge1$, i.e., $2d^2-1/2\ge0$.  Hence, our claim is equivalent to
the inequality $(2d^2-1/2)^2\ge(1/2-c^2-d^2)^2+c^2|z|^2$.  This can be
simplified to $d^2(3d^2-1)\ge(d^2-a^2-b^2)c^2$.  Therefore it suffices
to prove that $3d^2-1\ge c^2$.  In fact the stronger inequality
$3d^2-1\ge a^2$ holds.  Indeed, the inequality \eqref{eq:Nej(ii)}
implies that $d^2\ge a^2+b^2+c^2$.  By adding this inequality and the
inequality $d^2\ge a^2+|z|^2$, we obtain that $2d^2\ge1+a^2-d^2$. This
completes the proof of our claim.

By extracting square roots on both sides of $(2d^2-P)^2\ge
P^2-4|Q|^2$, we conclude that the inequality \eqref{ea:max(u3)} is
valid for all unit vectors $u_3$.  By invoking Eq. \eqref{eq:mu=max},
we obtain that $d\ge\mu(\psi)$.  This completes the proof of part (b),
and of the theorem.
\end{proof}

We derive two consequences of the above theorem.
\begin{corollary}\label{cr:Posledice}\
\begin{description}
\item[(i)] $\D$ is the closure of $\D^0$ (in particular, $\D^0\ne\emptyset$).
\item[(ii)] All points $p=(a,b,c,d,z)\in\D$ satisfy the inequality
\begin{eqnarray} \label{eq:StrongIneq}
2abc+d(2d^2-1)\ge0.
\end{eqnarray}
\end{description}
\end{corollary}
\begin{proof}
(i) It follows from Theorem \ref{thm:NejDelta} that
  $\D^0\ne\emptyset$. For instance, all inequalities defining $\D$ are
  strict at the point $p=(8,4,2,11,2+4i)/15$ and so $p\in\D^0$.  Now,
  the assertion follows from Lemma \ref{le:connection}.

(ii) In view of (i), it suffices to prove this inequality when
  $p\in\D^0$. If $\d$ denotes the left hand side of
  Eq. \eqref{eq:Nej(iv)}, then $\d+4abcdx^2$ is a polynomial in
  $a,b,c,d$ and $|z|^2$. After substituting $|z|^2=1-s_1-d^2$ in this
  polynomial, we obtain the inequality
  $(d(d^2-s_1)-2abc)(2abc+d(2d^2-1))\ge0$. As $p\in\D^0$ we have
  $d(d^2-s_1)-2abc>0$ and so $2abc+d(2d^2-1)\ge0$.
\end{proof}

\subsection{The boundary of $\D$}

We can now describe the boundary of $\D$.

\begin{proposition} \label{pp:Granica}
For $p=(a,b,c,d,z)\in\D$ with $z=x+iy$, we have $p\in\partial\D$ if and only if $cx(a-b)(b-c)\Phi(p)=0$.
\end{proposition}
\begin{proof}
Assume that $cx(a-b)(b-c)\Phi(p)=0$. If $cx(a-b)(b-c)=0$, it is
obvious that $p\in\partial\D$. Now let $cx(a-b)(b-c)\ne0$, i.e.,
$a>b>c>0$ and $x>0$. Then we have $\Phi(p)=0$. By computing the
gradient of $\Phi$, in the $6$-dimensional Euclidean space with
coordinates $a,b,c,d,x,y$, we find that at the point $p$ we have
$a\frac{\partial\Phi}{\partial b}-b\frac{\partial\Phi}{\partial
  a}=-2c(a^2-b^2)(4abc+d(x^2-y^2))$ and $x\frac{\partial\Phi}{\partial
  y}-y\frac{\partial\Phi}{\partial x}=8abcdxy$. Since at least one of
these two expressions is nonzero, the gradient $(\nabla\Phi)_p$ is not
parallel to $p$. Consequently, in any neighborhood of $p$ on the unit
sphere of $W_6$ the polynomial $\Phi$ takes both positive and negative
values. Hence $p\in\partial\D$ by Theorem \ref{thm:NejDelta}.

We shall prove the converse by contradiction. Thus, in addition to the
hypothesis $p\in\partial\D$, we shall assume that
$cx(a-b)(b-c)\Phi(p)\ne0$. Consequently, we have $a>b>c>0$, $x>0$. As
$d>0$ and $p\in\partial\D$, we infer that equality must hold in at
least one of the inequalities \eqref{eq:Nej(i)}, \eqref{eq:Nej(ii)},
\eqref{eq:Nej(iii)}.  But in each of these three cases, part (a) of
Theorem \ref{thm:NejDelta} asserts that equality implies that
$bx=0$. Hence, we have a contradiction.
\end{proof}

To demonstrate that the hypothesis $p\in\D$ cannot be omitted,
consider the point $p=(12,6,4,9,2,2)/2\sqrt{35}$ for which
$cx(a-b)(b-c)\Phi(p)>0$ holds, but $p$ is not in $\D$.

By using Lemma \ref{le:maximum}, it follows from \cite[Theorem
  1]{cxz10} that the minimum of $\mu'(\psi)$ over all unit vectors
$\ket{\psi}\in\we^3(V)$ is equal to $2/3$, and that the minimum is
attained only at the states $\ket{\psi}$ which, when regarded as a
three-qubit state, are LU-equivalent to the W state \cite{dvc00}. They
are all LU-equivalent to the point $p=(1/3,1/3,1/3,2/3,i\sqrt{2}/3)$
in $\D$. It is easy to compute the invariants $M_i$ at $p$:
\begin{eqnarray} \label{eq:inv-p}
M_1=1,~ M_2=2/3,~ M_3=1/9,~ M_4=4/27,~ M_6=8/729,~ M_5=M_7=0.
\end{eqnarray}
By using Lemma \ref{le:KanForm}, we deduce that the projection of $\D$
on the $d$-axis is the closed interval $[2/3,1]$. Thus if
$\ket{\psi}\in\D$ corresponds to the point $(a,b,c,d,z)$, then $2/3\le
d\le1$. If $d=1$ then $\ket{\psi}=e_{246}$, a decomposable
$3$-vector. On the other hand, if $d=2/3$ then it follows easily from
Eqs. \eqref{eq:M'1}--\eqref{eq:M'7} and \eqref{eq:inv-p} that
$a=b=c=1/3$ and $z=\pm i\sqrt{2}/3$. Thus there are exactly two points
$(a,b,c,d,z)\in\Delta$ with $d=2/3$. Note that these points belong to
$\partial\D$.

We shall prove now that the minimum of $\mu'(\psi)$ over all
normalized states $\ket{\psi}$ in $\we^3(V)$ is $2/3$. Thus, we obtain
an independent proof of the fact from \cite{cxz10} mentioned above.
\pagebreak[3]

\begin{proposition} \label{pp:Minimum}
Let $p=(a,b,c,d,z)\in\D$ where $z=x+iy$.
\begin{description}
\item[(i)] The minimum of $d$ over all points $p\in\D$ is $2/3$.
It is achieved only at the two points with coordinates
$a=b=c=1/3$, $d=2/3$ and $z=\pm i\sqrt{2}/3$.
\item[(ii)] $\min_\psi \mu(\psi)=2/3$, where the minimum is over all
  normalized states $\ket{\psi}\in\we^3(V)$.
\end{description}
\end{proposition}
\begin{proof}
(i) Let $p\in\D$ be any point where the minimum occurs. Clearly, we
  must have $p\in\partial\D$. Due to the examples given in the
  proposition, we have $d\le2/3$, and so $2d^2-1<0$. The inequality
  \eqref{eq:StrongIneq} implies that $c>0$.

Assume that $x>0$. As $bx>0$, the inequalities
\eqref{eq:Nej(i)}--\eqref{eq:Nej(iii)} must be strict at the point
$p$. Since $2d^2-1<0$ and $c>0$, the inequality \eqref{eq:Nej(iv)}
implies that $y\ne0$. Suppose now that the inequality
\eqref{eq:Nej(iv)} at the point $p$ is an equality, i.e., that
$\Phi(p)=0$. By solving the equations $M_1=1$ and $\Phi(p)=0$ for
$x^2$ and $y^2$, we obtain that
\begin{eqnarray} \label{eq:x2}
4abcdx^2 &=& (d(d^2-s_1)-2abc)(2abc+d(2d^2-1)), \\ \label{eq:y2}
4abcdy^2 &=& (d(d^2-s_1)+2abc)(2abc-d(2d^2-1)).
\end{eqnarray}
We choose a point $p'=(a',b',c',d',z')\in\D$, $z'=x'+iy'$, close to
$p$ such that $a'>a$, $b'=b$, $c'=c$, $d'<d$, and $\Phi(p')=0$.  We
can do that by simply setting $d'=d\cos\theta-a\sin\theta$ and
$a'=d\sin\theta+a\cos\theta$, where $\theta>0$ is small, and then
computing $x'$ and $y'$ from the above two equations with $a$ and $d$
replaced by $a'$ and $d'$, respectively. These two equations guarantee
that $\Phi(p')=0$. As $x>0$ and $y\ne0$, we have to choose $x'>0$, and
$y'$ to have the same sign as $y$. For sufficiently small $\theta>0$,
the inequalities \eqref{eq:Nej(i)}--\eqref{eq:Nej(iii)} will be
satisfied at the point $p'$, and we will have equality in
\eqref{eq:Nej(iv)}.  As $b'=b$ and $c'=c$, the other inequalities
defining $\D$ will also be satisfied.  This contradicts with the
hypothesis that $d$ takes the minimal value at the point $p$. We
conclude that also the inequality \eqref{eq:Nej(iv)} is strict at $p$.

Consider the function $f(t)=1-s_1-y^2-(x+t)^2$ of a real variable
$t$.  At the point $t=0$ we have $f(0)=d^2$. Note that $x<1$ and
$f'(0)=-2x<0$.  Hence we can choose a small $\varepsilon>0$ such that
$f(\varepsilon)<d^2$ and all four inequalities
\eqref{eq:Nej(i)}--\eqref{eq:Nej(iv)} are still valid at the point
$q=(a,b,c,\sqrt{f(\varepsilon)},x+\varepsilon+iy)$. By part (b) of
Theorem \ref{thm:NejDelta}, we conclude that $q\in\D$. This
contradicts the fact that the minimum of the coordinate $d$ over $\D$
occurs at the point $p$.  Thus we have shown that $x=0$.

By the inequalities \eqref{eq:Nej(ii)} and \eqref{eq:StrongIneq}, we
have $d(d^2-s_1)\ge2abc\ge d(1-2d^2)$, and so $c>0$ and
$s_1\le3d^2-1$. By the arithmetic-geometric mean inequality we have
\begin{eqnarray}
\left( \frac{3d^2-1}{3} \right)^3 \ge
\left( \frac{s_1}{3} \right)^3 \ge (abc)^2 \ge
\frac{d^2}{4} (1-2d^2)^2.
\end{eqnarray}
By expanding the leftmost and rightmost member, we obtain that
$d\ge2/3$.  Moreover, if $d=2/3$ then the above inequalities
become equalities, and so we must have $a=b=c$. By (i) we have
$a\ge1/3$ and from \eqref{eq:Nej(ii)} we have $a\le1/3$. Hence,
$p=(1,1,1,2,\pm i\sqrt{2})/3$.

(ii) Since $\mu'$ is LU-invariant, by Lemma \ref{le:KanForm}
we may minimize over $\D$ only. By Lemma \ref{le:maximum} we
have $\mu'=\mu$ on $\D$. Hence (ii) follows from (i).
\end{proof}

\section{Three-fermion canonical form}
\label{sec:canonical}

We begin by introducing the notation $p\mapsto p'$ for the projection
map $\bR^4\times\bC\to\bR^4$.  Thus for any point
$p=(a,b,c,d,z)\in\bR^4\times\bC$, we set $p'=(a,b,c,d)$. Let $\D'$ be
the image of $\D$ under this projection map, i.e.,
$\D'=\{p':p\in\D\}$.

\subsection{The canonical form}
\label{sec:canonicalform}

In order to state the main result of this section, the canonical form
for fermionic states, we have to resolve the question of
LU-equivalence of points in $\D$. The most important fact is that two
distinct points in $\D$ which are LU-equivalent must lie on
$\partial\D$.

\begin{proposition}\label{pp:Delta}
Let $p=(a,b,c,d,z)$ and $q=(\tilde{a},\tilde{b},\tilde{c},\tilde{d},\tilde{z})$ be distinct points of
$\D$, and let $z=x+iy$ and $\tilde{z}=\tilde{x}+i\tilde{y}$. Then $p$ and $q$ are
LU-equivalent if and only if the following two conditions hold
\begin{description}
\item[(i)] $p'=q'$ and also $\tilde{x}=x$ if $c>0$;
\item[(ii)] $cx\Phi(p)=\tilde{c}\tilde{x}\Phi(q)=0$.
\end{description}
In particular, if $p$ and $q$ are LU-equivalent, then $p,q\in\partial\D$.
\end{proposition}
\begin{proof}
{\em Necessity.} Since $p$ and $q$ are LU-equivalent, they have the
same values of the invariants $M_i$, $i=1,\ldots,7$. It follows from
Definition \ref{def:Delta} (ii) that $\tilde{d}=d$. By substituting $|z|^2$
with $t$ on the left hand side of Eq. \eqref{eq:j1}, we obtain a cubic
polynomial $g(t)$. As $|z|^2$ is one of its roots, $g(t)$ factorizes
as $g(t)=3d^2(t-|z|^2)h(t)$ where
\begin{eqnarray}
h(t)=d^2t^2-d^2(3-6d^2-|z|^2)t+d^2(2d^2-1)(4d^2-2+|z|^2)-8s_3.
\end{eqnarray}
As the discriminant of $h(t)$ is nonnegative, $h(t)$ has two real
roots $t_1\le t_2$. As $t_1+t_2=3-6d^2-|z|^2$, we have
$2|z|^2-(t_1+t_2)=6d^2+3|z|^2-3=3(d^2-s_1)\ge0$ by Theorem
\ref{thm:NejDelta}. Hence, $|z|^2\ge(t_1+t_2)/2$. On the other hand,
$h(|z|^2)=2\left( d^2(d^2-s_1)^2 -4s_3 \right)$ is nonnegative by the
same theorem. It follows that $|z|^2\ge t_2$. Thus all roots of $g(t)$
are real, and $|z|^2$ is the largest root. By the same argument,
$|\tilde{z}|^2$ is the largest root of $g(t)$, and so $|\tilde{z}|=|z|$.

Eqs. \eqref{eq:s1f}--\eqref{eq:s3f} imply that the elementary
symmetric functions $s_1,s_2,s_3$ of $a,b,c$ are the same as those of
$\tilde{a},\tilde{b},\tilde{c}$.  As $a\ge b\ge c\ge0$ and
$\tilde{a}\ge\tilde{b}\ge\tilde{c}\ge0$, it follows that
$\tilde{a}=a$, $\tilde{b}=b$ and $\tilde{c}=c$.  Consequently, we have
$p'=q'$. If $c>0$, then the inequalities $x\ge0$, $\tilde{x}\ge0$, and
Eq. \eqref{eq:M'3} imply that $\tilde{x}=x$. Thus (i) holds.

Assume that $cx\ne0$. By (i) we have $\tilde{x}=x$ and
$\tilde{y}=-y\ne0$. Since $M_7$ has the same value at $p$ and $q$ and
$\tilde{a}\tilde{b}\tilde{c}\tilde{d}\tilde{x}\tilde{y}=-abcdxy\ne0$,
Eq. \eqref{eq:M'7} implies that
$\Phi(q)=-\Phi(p)$. By Theorem \ref{thm:NejDelta} both $\Phi(p)$ and
$\Phi(q)$ are nonnegative, and so $\Phi(q)=\Phi(p)=0$.  Hence, (ii)
holds.

{\em Sufficiency.} We have to show that $M_i(p)=M_i(q)$ for
$i=1,\ldots,7$. This follows immediately by inspection of the formulae
\eqref{eq:M'1}--\eqref{eq:M'7}.

Finally, note that if $p$ and $q$ are LU-equivalent, then (ii) and
Lemma \ref{pp:Granica} imply that $p,q\in\partial\D$.
\end{proof}

Let us denote by $\cO_\psi$ the LU-orbit of $\ket{\psi}\in\we^3
(V)$. Thus,
$\cO_\psi=\Un(6)\cdot\ket{\psi}=\{g\cdot\ket{\psi}\colon g\in\Un(6)\}$. Assume
that $\|\psi\|=1$. The intersection $\cO_\psi\cap\D$ consists of a
single point if and only if one of the following holds:
\begin{description}
\item[(i)] $cx\Phi(p)>0$;
\item[(ii)] $c>0$ and $x\Phi(p)=y=0$.
\item[(iii)] $c=z=0$.
\end{description}
Otherwise, $\cO_\psi\cap\D$ is either
\begin{description}
\item[(iv)] a pair of points: $\{p,q\}$ where $p=(a,b,c,d,z)$,
$q=(a,b,c,d,z^*)$, $cy>0$, and $x\Phi(p)=0$;
\\
or
\item[(v)] a semicircle: $\{(a,b,0,d,re^{it}):|t|\le\pi/2\}$ where
$r=\sqrt{1-a^2-b^2-d^2}>0$.
\end{description}
While the case (i) covers all points in $\D^0$ and some points on
$\partial\D$, in all other cases the points lie on $\partial\D$. It is
easy to see that all five cases indeed occur.

We now state our main result, which follows immediately from Lemma
\ref{le:KanForm} and Proposition \ref{pp:Delta}.
\begin{theorem}\label{thm:KanForma}
Any normalized pure fermionic state $\ket{\varphi}\in\we^3(V)$,
with $\dim V=6$, is LU-equivalent to a state
\begin{eqnarray}
\ket{\psi}=ae_{235}+be_{145}+ce_{136}+de_{246}+ze_{135},
\end{eqnarray}
where $p:=(a,b,c,d,z)\in\D$, $z=x+iy$. Such a state $\ket{\psi}$ is
unique if $cx\Phi(p)>0$; in particular this is true if $p\in\D^0$. To
guarantee the uniqueness of $\ket{\psi}$ when $cx\Phi(p)=0$, we require
that (i) $y\ge0$ and (ii) $y=0$ if $c=0$.

(For a simpler description of the canonical region $\D$ see
Proposition \ref{pp:projection} below.)
\end{theorem}
Our next objective is to give a more geometric description of the
canonical region $\D$ and its projection $\D'$. By restricting the
projection map $\bR^4\times\bC\to\bR^4$, we obtain the map $\D\to\D'$,
$p\mapsto p'$. The fibre of this map over a point $p':=(a,b,c,d)\in\D'$ is
the set $F_{p'}:=\{q\in\D\colon q'=p'\}$. We shall determine the nature of
these fibres and find explicit inequalities which define $\D'$ as a
subset of $\bR^4$.

For convenience, we set $r:=\sqrt{1-a^2-b^2-c^2-d^2}$ and denote by
$S_{p'}$ the semicircle consisting of all points $(p',z)=(a,b,c,d,z)$
with $z=x+iy$ such that $|z|=r$ and $x\ge0$. Note that
$F_{p'}=S_{p'}\cap\D\supseteq\{p\}$ for all $p\in\D$. If $c=0$ it is
easy to see that $F_{p'}=S_{p'}$.  Note that $S_{p'}=\{p\}$ if and
only if $r=0$.

We claim that if $F_{p'}=\{p\}$, then $r=0$. Indeed, observe that
$p=(p',z)\in\D$, with $z=x+iy$, implies that $(p',z^*)\in\D$, and so we
must have $y=0$. Let $q:=(p',ir)\in S_{p'}$. Since $p\in\D$, the
inequalities \eqref{eq:Nej(i)}--\eqref{eq:Nej(iv)} are valid at
$p$. The first two do not involve $z$ and remain valid at $q$. By
using the fact that $y=0$, we see that the remaining two inequalities
remain valid when we replace $p$ by $q$. Hence, Theorem
\ref{thm:NejDelta} implies that $q\in\D$. Consequently, $q\in F_{p'}$,
and so we must have $q=p$, i.e., $r=0$.

Finally, let us denote by $\Psi(a,b,c,d,z)$ the left hand side of
\eqref{eq:Nej(iii)}.  It is evident from the definitions of $\Phi$ and
$\Psi$ that $\Phi(a,b,c,d,re^{it})$ and $\Psi(a,b,c,d,re^{it})$,
considered as functions of $t\in[0,\pi/2]$, are constant if $c=0$, and
are strictly increasing if $c>0$.

\begin{lemma}\label{le:fibre}
Let $p=(a,b,c,d,re^{i\theta})\in\D$, $r>0$, $|\theta|<\pi/2$.
The following assertions hold:
\begin{description}
\item[(i)] If $\Phi(p',r)\ge0$ then $F_{p'}=S_{p'}$.
\item[(ii)] If $\Phi(p',r)<0$ then $F_{p'}$ consists of all points
  $(p',z')\in S_{p'}$, $z'=x'+iy'$, such that $0\le x'\le x_0$ where
  $x_0$ is the nonnegative solution of Eq. \eqref{eq:x2}.
\end{description}
\end{lemma}
\begin{proof}
Without any loss of generality, we may assume that $\theta\ge0$.  For
$t\in[0,\pi/2]$, let $p(t)=(p',re^{it})$. If $c=0$, then the functions
$\Phi(p(t))$ and $\Psi(p(t))$ are constant and so (i) holds in that
case.  We assume that $c>0$. The inequalities \eqref{eq:Nej(i)} and
\eqref{eq:Nej(ii)} are valid at all points of the fibre $F_{p'}$, not
only at the point $p$, because their left hand sides do not depend on
$z$. Moreover, since $b>0$ these two inequalities are strict at each
point $p(t)\in F_{p'}$ with $0\le t<\pi/2$.

(i) We have to prove that $p(t)\in\D$ for $t\in[0,\pi/2]$. The proof
is based on part (a) of Theorem \ref{thm:NejDelta}. As $p\in\D$, we
have $\Psi(p(\theta))=\Psi(p)\ge0$. Assume that $\Psi(p(0))<0$ and so
$p(0)\notin\D$. There is a unique $t_1\in(0,\theta]$ such that
  $\Psi(p(t_1))=0$. As $\Phi(p(t_1))>0$, we have $p(t_1)\in\D$. Since
  $b>0$ and $0<t_1<\pi/2$, we must have $\Psi(p(t_1))>0$ by part (a)
  of Theorem \ref{thm:NejDelta}. Hence, we have a contradiction. We
  conclude that $\Psi(p(0))\ge0$. It follows that for any
  $t\in[0,\pi/2]$, we have $\Phi(p(t))\ge0$ and $\Psi(p(t))\ge0$, and
  so $p(t)\in\D$.

(ii) Since $p\in\D$, we have $\Phi(p)\ge0$ and $\Psi(p)\ge0$. On the
  other hand, by the hypothesis, we have $\Phi(p',r)<0$. It follows
  that there is a unique $t_1\in(0,\theta]$ such that
    $\Phi(p(t_1))=0$. We deduce that $x_0=r\cos t_1$, and that
    $\Phi(p(t))\ge0$ if and only if $t_1\le t\le\pi/2$. If
    $\Psi(p(t_1))<0$, we can deduce a contradiction by the same
    argument as in part (i). Therefore $\Psi(p(t_1))\ge0$, and the
    assertion (ii) follows easily.
\end{proof}

We now determine the projection $\D'$ of $\D$ and simplify the set of
inequalities in Theorem \ref{thm:NejDelta} which define $\D$.
\begin{proposition} \label{pp:projection}\
\begin{description}
\item[(i)] The subset $\D'\subseteq\bR^4$ consists of all points
  $p'=(a,b,c,d)$ which satisfy the inequalities $a\ge b\ge c\ge0$,
  $d>0$, $a^2+b^2+c^2+d^2\le1$ as well as the inequalities \eqref{eq:Nej(ii)}
  and \eqref{eq:StrongIneq}.
\item[(ii)] A point $p=(a,b,c,d,z)\in\bR^4\times\bC$ belongs to $\D$
  if and only if $a^2+b^2+c^2+d^2+|z|^2=1$ and $p$ satisfies the following
  inequalities (with $s_1=a^2+b^2+c^2$)
\begin{eqnarray}
&& a\ge b\ge c\ge 0,~x\ge0,~d>0; \\
&& d(d^2-s_1)-2abc\ge0; \label{eq:Nej(ii)'},  \\
&& 2abc-d(1-2d^2)\ge0; \label{eq:StrongIneq'} \\
&& 4abcdx^2\le \left( d(d^2-s_1)-2abc \right)
\left( 2abc-d(1-2d^2) \right). \label{eq:Nej(iv)'}
\end{eqnarray}
\end{description}
\end{proposition}
\begin{proof}
(i) {\em Necessity.} Let $p':=(a,b,c,d)\in\D'$ and let
$r=\sqrt{1-s_1-d^2}$. It follows from Lemma \ref{le:fibre} that
the point $p:=(p',ir)\in F_{p'}$. Consequently, $p'$ must satisfy all
inequalities listed in the proposition.

{\em Sufficiency.}  Let $p=(a,b,c,d,ir)$, where
$r=\sqrt{1-s_1-d^2}\ge0$. It suffices to show that $p\in\D$. By
Theorem \ref{thm:NejDelta}, it suffices to verify that $p$ satisfies
the inequalities \eqref{eq:Nej(i)}, \eqref{eq:Nej(iii)} and
\eqref{eq:Nej(iv)}. As $d>0$, the inequality \eqref{eq:Nej(i)} follows
immediately from \eqref{eq:StrongIneq'}.  One can easily verify that
\begin{eqnarray}
\Psi(p)-\Phi(p)=(ab-cd)^2(3d^2-1-s_1)+
(d(d^2-s_1)+2abc)(d(2d^2-1)+2abc).
\end{eqnarray}
The inequalities \eqref{eq:Nej(ii)'} and \eqref{eq:StrongIneq'} imply
that $3d^2-1-s_1\ge0$, and so $\Psi(p)-\Phi(p)\ge0$. Since at the
point $p$ we have $z=ir$ and so $x=0$, we obtain that
$\Phi(p)=(d(d^2-s_1)-2abc)(2abc+d(2d^2-1))$. Thus, the inequalities
\eqref{eq:Nej(ii)'} and \eqref{eq:StrongIneq'} also imply that
$\Phi(p)\ge0$. Thus $\Psi(p)\ge\Phi(p)\ge0$, and so the inequalities
\eqref{eq:Nej(iii)} and \eqref{eq:Nej(iv)} also hold at the point $p$.

(ii) We already know that the listed inequalities are necessary.
(Note that the inequalities \eqref{eq:Nej(iv)} and \eqref{eq:Nej(iv)'}
are equivalent to each other.)  To prove the sufficiency, we observe
that all inequalities listed in part (i) are satisfied at the point
$p':=(a,b,c,d)$, and so $p'\in\D'$. Since $|z|=r:=\sqrt{1-s_1-d^2}$
and $x\ge0$, we have $z\in S_{p'}$, i.e., $z=re^{it}$ for some real
$t$, $|t|\le\pi/2$. As $p$ satisfies the inequality
\eqref{eq:Nej(iv)'}, we have $\Phi(p',re^{it})=\Phi(p)\ge0$. Now Lemma
\ref{le:fibre} implies that $p\in F_{p'}$ and, in particular,
$p\in\D$.
\end{proof}

We point out that none of the four inequalities $d>0$, and
\eqref{eq:Nej(ii)'}--\eqref{eq:Nej(iv)'} in part (ii) can be
omitted. For the inequality $d>0$ we can use the point
$p=(3,0,0,0,4)/5$ which is not in $\D$ and satisfies the other three
inequalities.  Similarly, for the remaining three inequalities:
$(1,1,1,1,0)/2$ violates only \eqref{eq:Nej(ii)'},
$(1,0,0,1,1)/\sqrt{3}$ violates only \eqref{eq:StrongIneq'}, and
$(4,2,2,6,2+3i)/\sqrt{73}$ violates only \eqref{eq:Nej(iv)'}.

Note also that in view of Proposition \ref{pp:Minimum}, the inequality
$d>0$ can be replaced by the stronger inequality $d\ge2/3$.

\subsection{\label{subsec:algorithm} The algorithm}

We shall describe a numerical algorithm which, for a generic
normalized state $\ket{\phi}\in\we^3(V)$ as input, computes its
canonical form $\ket{\psi}$. Let us denote by $p=(a,b,c,d,z)$ the
point in $\D$ that corresponds to $\ket{\psi}$. So, our problem is to
compute the coordinates $a,b,c,d$ and $z=x+iy$.  Since $\ket{\phi}$ is
given, we can compute the values of the invariants $M_i$ at
$\ket{\phi}$. For simplicity, we shall write in this subsection
$M_i=M_i(\phi)$. Since $\ket{\psi}$ is LU-equivalent to $\ket{\phi}$,
we also have $M_i(p)=M_i$ for $i=1,\ldots,7$. Thus, $p=(a,b,c,d,z)$ is
a solution of this system of equations. In order to solve numerically
this system of equations, we proceed as follows.

By eliminating the expression $x^2-y^2$ from Eqs. \eqref{eq:M'3} and
\eqref{eq:M'5} we obtain
\begin{eqnarray}
3M_5+4d^2M_3 &=& 3d^4|z|^4 +24(s_3-s_2d^2)d^2
                +4(s_2+s_1d^2)d^2.        \label{eq:j5}
\end{eqnarray}
By solving Eqs. \eqref{eq:M'1}, \eqref{eq:M'2} and
\eqref{eq:j5} for $s_1$, $s_2$, and $s_3$, we obtain that
\begin{eqnarray}
s_1 &=& 1 -d^2 -|z|^2, \label{eq:s1f} \\
2s_2 &=& 2d^4 -(2+|z|^2)d^2 +M_2, \label{eq:s2f} \\
24d^2s_3 &=& 24d^8 -12( 2+|z|^2)d^6
+3( 4M_2 +2|z|^2 -|z|^4 )d^4 \notag \\
&&\quad +2( 2M_3-M_2 )d^2 +3M_5. \label{eq:s3f}
\end{eqnarray}
As $d>0$, by substituting the above expressions into
Eqs. \eqref{eq:M'4} and \eqref{eq:M'6} we obtain the equations
\begin{eqnarray}
&& 3d^4|z|^6 +9(2d^2-1)d^4|z|^4
   +(6(1-2M_2)d^4+2(M_2-2M_3)d^2-3M_5)|z|^2 \notag \\
&& \quad-96d^{10}+144d^8-48(1+M_2)d^6+16(2M_2-M_3)d^4
 \notag \\ \label{eq:j1}
&& \quad+2(2M_3+3M_2^2-M_2-12M_4-6M_5)d^2+3M_5=0, \\
&& 27d^8|z|^8 +36d^6(8d^4-4d^2+M_2)|z|^6 \notag \\
&& \quad+6d^4(144d^8-144d^6+6d^4(7+4M_2)-4d^2(4M_2+M_3)
-3M_5)|z|^4 \notag \\
&& \quad+12d^2(4d^4-2d^2+M_2)(6d^4(1-2M_2)
+2d^2(M_2-2M_3)-3M_5)|z|^2 \notag \\
&& \quad-2304d^{16}+4608d^{14}-576(5+4M_2)d^{12}
+192(3+16M_2-2M_3)d^{10} \notag \\
&& \quad+96(4M_3-11M_2-6M_2^2-3M_5)d^8
+48(M_2-2M_3+8M_2^2+6M_5-4M_2M_3)d^6 \notag \\
&& \quad+4(16M_2M_3-7M_2^2-18M_5-36M_2M_5-4M_3^2
-144M_6)d^4 \notag \\  \label{eq:j2}
&& \quad+24(2M_2-M_3)M_5d^2-9M_5^2=0.
\end{eqnarray}

By substituting formally $|z|^2$ and $d^2$ in these two equations with
new indeterminates $s$ and $t$, respectively, we obtain two
polynomials $g(s,t)$ and $h(s,t)$ which have a common solution, namely
$s=|z|^2$, $t=d^2$. We can eliminate now $s$ from these two equations
by computing the resultant of the polynomials $g(s,t)$ and $h(s,t)$
with respect to the variable $s$. This resultant has the form $-2^4
3^5 t^{12} f(t)$, where $f(t)$ is a polynomial in $t$ of degree 8
whose coefficients are polynomials in the six parameters
$M_1,\ldots,M_6$ with integer coefficients. As the fully expanded
$f(t)$ has 550 terms, it is only given in Appendix
\ref{dodatak}.  Since we know that $t=d^2$ is a root of $f(t)$, the
polynomial $f(t)$ has a factorization $f(t)=(t-d^2)f_1(t)$. It turns
out that $f_1(t)$, when expanded as a polynomial in $t$ and the
parameters $M_i$, has 41967 terms.

After substituting the expressions \eqref{eq:M'1}--\eqref{eq:M'6} into
$f_1(t)$, by using Maple we obtain the following factorization
\begin{eqnarray} \label{eq:f_1}
f_1(d^2) &=& 2^8 3^5 (abcd)^2(d^2-a^2)(d^2-b^2)(d^2-c^2)
\left( d^2(d^2-s_1)^2-4s_3 \right)^3 \cdot\Phi(p).
\end{eqnarray}
Hence, $d^2$ is a multiple root of $f(t)$ if and only if $c=0$ or
$d(d^2-s_1)-2abc=0$ or $\Phi(p)=0$. Recall that for $p\in\D$ the
equality $d(d^2-s_1)-2abc=0$ implies that $bx=0$. It follows that
$d^2$ is a multiple root of $f(t)$ if and only if $cx\Phi(p)=0$.

Computing $d^2$ is the main step in the algorithm. We applied our
algorithm to about 100 randomly chosen input states $\ket{\phi}$. In
all these cases the equation $f(t)=0$ had $4$, $6$, or $8$ real roots,
all of them lying in the interval $(0,1)$. Moreover, $d^2$ was always
the largest of these real roots. Since $p\in\D$, the correct value of
$d$ will be equal to $\mu(\phi)$.

Once the value of $d^2$ is selected, the remaining steps of the
algorithm are straightforward. We compute $|z|^2$ by solving the
polynomial equation $g(s,d^2)=0$. As shown in the proof of Proposition
\ref{pp:Delta}, this cubic has three real roots and $|z|^2$ is its
largest root. Then the equations \eqref{eq:s1f}--\eqref{eq:s3f} provide
the values of $s_1$, $s_2$, and $s_3$. As $a^2$, $b^2$, $c^2$ are the
roots of the cubic $t^3-s_1t^2+s_2t-s_3$, and we know that $a\ge b\ge
c\ge0$, we can compute $a$, $b$, and $c$. Next, $x^2$ (and $y^2$) can be
computed from Eq. \eqref{eq:M'3}. As $x\ge0$, we obtain also the value
of $x$. Finally, we can determine the sign of $y$ by using
Eq. \eqref{eq:M'7}.

When applying the algorithm, we had to increase the precision by
setting the global Maple parameter \texttt{Digits} from $10$ (the
default) to $200$ since some of the coefficients in $f(t)=0$ are small
and hence the equation might not be well conditioned.

We conjecture that, if $f(t)$ is not identically zero, then\\
(1) all real roots of $f(t)$ belong to the interval $[0,1]$ and\\
(2) $d^2$ is the largest real root of $f(t)$.

\section{Real and quasi-real pure fermionic states}
\label{sec:quasireal}

We say that a state $\ket{\psi}\in\we^3(V)$ is {\em real} if there
exists $g\in\Un(6)$ such that all components of $g\cdot\ket{\psi}$,
with respect to the standard basis $\{e_{ijk}\}$, are real. Let us
denote by $\cO_\psi$ the LU-orbit containing the state $\ket{\psi}$,
i.e., $\cO_\psi=\Un(6)\cdot\ket{\psi}=\{g\cdot\ket{\psi}\colon
g\in\Un(6)\}$.  By $\ket{\psi^*}$ we denote the complex conjugate of
$\ket{\psi}$ computed in the standard basis. We say that $\ket{\psi}$
is {\em quasi-real} if $\ket{\psi^*}$ and $\ket{\psi}$ are
LU-equivalent. It is immediate that $\ket{\psi}$ is {\em quasi-real}
if and only if $\cO_{\psi^*}=\cO_\psi$.  Note that every real state is
quasi-real.

One defines similarly the real and quasi-real pure states of three
qubits.  A pure three-qubit state is real if in some orthonormal basis
all of its components are real numbers.  It is quasi-real if it is
LU-equivalent to its complex conjugate state. It was proved in
\cite[Eq. (34) and Appendix B]{ajt01} that the pure three-qubit
quasi-real states are real.

By using the invariants, we can characterize the quasi-real
states.
\begin{proposition}\label{pp:quasi-real}
A state $\ket{\psi}\in\wedge^3(V)$ is quasi-real if and only if
$M_7(\psi)=0$.
\end{proposition}
\begin{proof}
Clearly we may assume that $\|\psi\|=1$. Since $M_7$ is a unitary
invariant, we may replace $\ket{\psi}$ with any state in the orbit
$\cO_\psi$. Hence, by Lemma \ref{le:KanForm} we may assume that
$\ket{\psi}\in\Delta$ and that it is given by
Eq. \eqref{eq:5dimVektor}. Thus $\ket{\psi}$ corresponds to the point
$(a,b,c,d,z)\in\Delta$ and $\ket{\psi^*}$ corresponds to the point
$(a,b,c,d,z^*)\in\Delta$. It is obvious from Eqs.
\eqref{eq:M'1}--\eqref{eq:M'7} that $M_i(\psi)=M_i(\psi^*)$ for
$i=1,\ldots,6$ and $M_7(\psi^*)=-M_7(\psi)$. Now the assertion follows
from Corollary \ref{cr:Equivalence}.
\end{proof}

We shall need the following simple lemma.
\begin{lemma}\label{le:DetDx}
Let $p=(a,b,c,d,z)\in\D$, $c>0$, and let $D_a$, $D_b$, $D_c$ be
defined as in Sec. \ref{sec:RDM}. Then we have $D_c\ge D_b\ge
D_a>0$. Moreover, $D_a=D_b$ if and only if $a=b$, and $D_b=D_c$ if and
only if $b=c$. Consequently, the determinants $D_a$, $D_b$, $D_c$ are
uniquely determined by the invariants $M_2$, $M_4$, $M_6$.
\end{lemma}
\begin{proof}
As $c>0$, we have $d>a\ge b\ge c>0$. The assertions of the lemma follow
from the formulae $D_c-D_b=(d^2-a^2)(b^2-c^2)$ and
$D_b-D_a=(d^2-c^2)(a^2-b^2)$, together with Eqs. \eqref{eq:D-2}--\eqref{eq:D-6}.
\end{proof}

We can now extend the above mentioned result of Acin et
al. \cite{ajt01} to the fermionic states in $\wedge^3(V)$.
\begin{proposition} \label{pp:quasireal=extension}
Any quasi-real state $\ket{\psi}\in\wedge^3(V)$ is real.
\end{proposition}
\begin{proof}
By Lemma \ref{le:KanForm}, we may assume that $\ket{\psi}\in\D$.
Let $p=(a,b,c,d,z)$ be the point corresponding to $\ket{\psi}$.
If $c=0$, then we may assume that $z=x\ge0$ and so $\ket{\psi}$
is real. Thus we may assume that $c>0$.

As $\ket{\psi}$ is quasi-real, we have $\ket{\psi^*}\in\cO_\psi$. By
Theorem \ref{th:equiv}, the intersection $W\cap\cO_\psi$ is a single
$G$-orbit. Since both $\ket{\psi}$ and $\ket{\psi^*}$ belong to this
intersection, there exists $g\in G$ such that
$\ket{\psi^*}=g\cdot\ket{\psi}$. By using that $g=U\s$ for some
$U\in\Un(2)\times\Un(2)\times\Un(2)$ and some $\s\in S_3$, we obtain
that $\s^{-1}U\s\cdot\ket{\psi}=\s^{-1}\cdot\ket{\psi^*}$.  Note that
$\s^{-1}U\s\in\Un(2)\times\Un(2)\times\Un(2)$.  Since the invariants
$M_2$, $M_4$ ,$M_6$ take the same values at $\ket{\psi}$ and
$\ket{\psi^*}$, Lemma \ref{le:DetDx} implies that
$\s^{-1}\cdot\ket{\psi^*}=\ket{\psi^*}$. We can now apply the above
mentioned result of Ref. \cite{ajt01} to conclude the proof.
\end{proof}

\section{A new canonical form for pure three-qubit states}
\label{sec:threequbit}

A canonical form for pure three-qubit states was constructed first by
Acin et al. in \cite{ajt01}. Their Fig. 1 shows three inequivalent set
of states (I), (II), and (III) that can be used to construct a
canonical form.  Their canonical form is constructed by using the set
(III). They also show that case (II) can be reduced to case
(III). However, they were not able to handle the symmetric
decomposition (I).  By using our canonical form for three fermions, we
can obtain a solution in that case.

Recall that $\S$ denotes the unit sphere of $W_6$ defined by the
equation $a^2+b^2+c^2+d^2+|z|^2=1$.  We begin by introducing the region
\begin{eqnarray}
 \label{eq:Theta}
\T=\bigcup_{\s\in S_3} \s\cdot\D\subseteq\S,
\end{eqnarray}
which is obviously invariant under the action of $S_3$ (for the latter see Sec. \ref{sec:RDM}).

Since the inequalities \eqref{eq:Nej(ii)'}, \eqref{eq:StrongIneq'},
and \eqref{eq:Nej(iv)'} are not affected by the permutations of the
variables $a,b,c$, it follows easily that $\T$ is the region of $\S$
defined by these three inequalities and the linear inequalities
$a,b,c,x\ge0$ and $d>0$.  One can also show that the relative interior
$\T^0$ of $\T$ consists of the points of $\T$ at which all the
inequalities defining $\T$ are strict. It follows from Theorem
\ref{thm:NejDelta} that if at some point $p=(a,b,c,d,z)\in\T$ the
equality holds in \eqref{eq:Nej(ii)'} or \eqref{eq:StrongIneq'}, then
we must have $abcx=0$. Consequently, the relative boundary
$\partial\T$ of $\T$ consists of the points $p=(a,b,c,d,z)\in\T$ for
which $abcx\Phi(p)=0$.

Let us give yet another description of the region $\T$. Recall that
$\S_i$ denotes the unit sphere in the subspace $V_i\subseteq V$.
\begin{proposition}\label{pp:Theta}
The set $\T$ consists of all points $p=(a,b,c,d,z)\in\S$,
$z=x+iy$, satisfying the inequalities $a,b,c,x\ge0$ and such that
\begin{eqnarray} \label{eq:Theta=d}
d=\max_{\a,\b,\g} \abs{\braket{\a,\b,\g}{\ps}}, \quad
(\a,\b,\g)\in\S_1\times\S_2\times\S_3.
\end{eqnarray}
\end{proposition}
\begin{proof}
If $p=(a,b,c,d,z)\in\T$, then $p\in\s\cdot\D$ for some $\s\in S_3$. As
$S_3$ only permutes the coordinates $a,b,c$ and leaves the coordinates
$d$ and $z$ unchanged, we have $a,b,c,x\ge0$. Now \eqref{eq:Theta=d}
follows form Lemma \ref{le:maximum} and Definition \ref{def:Delta}.
\end{proof}

Evidently, two LU-equivalent three-qubit states in $\T$ have the same
coefficient $d$.

To prove the next proposition, we shall use the invariants
$Q_i$, $1\le i\le7$, of three qubits which are described in
Sec. \ref{sec:Inv3Qub}. At a point $p=(a,b,c,d,z)\in W_6$,
we have
\begin{eqnarray}
Q_2-Q_3 &=& 2(a^2-b^2)(d^2-c^2), \label{eq:Q2-3} \\
Q_3-Q_4 &=& 2(b^2-c^2)(d^2-a^2), \label{eq:Q3-4} \\
Q_4-Q_2 &=& 2(c^2-a^2)(d^2-b^2). \label{eq:Q4-2}
\end{eqnarray}
The following proposition gives the canonical form for pure
three-qubit states that was mentioned in \cite{ajt01} but left open.
\begin{proposition}\label{pp:ThetaKAN}
Any normalized pure three-qubit state $\ket{\phi}$ is
LU-equivalent to a state
\begin{eqnarray}
\ket{\psi} &=& a\ket{100}+b\ket{010}+c\ket{001}+d\ket{111}
+z\ket{000},
\end{eqnarray}
where $p:=(a,b,c,d,z)\in\T$, $z=x+iy$. Such state is unique if
$abcx\Phi(p)>0$, i.e., if $p\in\T^0$. To guarantee the uniqueness of
$\ket{\psi}$ when $abcx\Phi(p)=0$, i.e., when $p\in\partial\T$, we
require that (i) $y\ge0$ and (ii) $y=0$ if $abc=0$.
\end{proposition}
\begin{proof}
First we show that $\ket{\phi}$ is LU-equivalent to some state in
$\T$.  By Lemma \ref{le:KanForm} we can transform $\ket{\phi}$ into
$\D$ by the group $G$, i.e., there exist $\s\in S_3$ and
$g\in\Un(2)\times\Un(2)\times\Un(2)$ such that $\ket{\chi}:=\s
g\cdot\ket{\phi}\in\D$. Consequently, $\ket{\psi}:=g\cdot\ket{\phi}\in
\s^{-1}\cdot\D\subseteq\T$.

In order to prove the uniqueness assertions, suppose that we have two
distinct points $p=(a,b,c,d,z)$ and
$q=(\tilde{a},\tilde{b},\tilde{c},\tilde{d},\tilde{z})$ in $\T$, which
are equivalent under the local unitary group
$\Un(2)\times\Un(2)\times\Un(2)$ of three qubits. By definition of
$\T$, there exist $\s,\s'\in S_3$ such that $\s\cdot p,\tilde{\s}\cdot
q\in\D$. Note that the points $\s\cdot p$ and $\tilde{\s}\cdot q$
belong to the same orbit of $G$. Since $S_3$ only permutes the
coordinates $a,b,c$ and leaves the coordinates $d$ and $z$ unchanged,
Proposition \ref{pp:Delta} implies that $\tilde{d}=d$ and that
$(\tilde{a},\tilde{b},\tilde{c})$ is a permutation of $(a,b,c)$. Since
$Q_i(p)=Q_i(q)$ for all $i$, it follows from
Eqs. \eqref{eq:Q2-3}--\eqref{eq:Q4-2} that $\tilde{a}=a$,
$\tilde{b}=b$, $\tilde{c}=c$ and consequently $|\tilde{z}|=|z|$. As
usual, we write $z=x+iy$ and $\tilde{z}=\tilde{x}+i\tilde{y}$.  If
$abc>0$, then Proposition \ref{pp:Delta} also implies that $\tilde{x}=x$, and
so $\tilde{y}=-y\ne0$ because $p\ne q$. If $abc=0$, then the invariants
$Q_i(p)$ depend only on $a,b,c,d$ and $|z|^2$, and so all points
$(a,b,c,d,re^{it})$, where $r=\sqrt{1-s_1-d^2}$ and $|t|\le\pi/2$, are
LU-equivalent to each other. All uniqueness assertions easily
follow.
\end{proof}

\section*{Acknowledgments}
We thank Jianxin Chen, Zhengfeng Ji, and Mary Beth Ruskai for helpful
discussion on this paper.  LC was mainly supported by MITACS and
NSERC. The CQT is funded by the Singapore MoE and the NRF as part of
the Research Centres of Excellence programme.  DD was supported in
part by an NSERC Discovery Grant. BZ is supported by NSERC and
CIFAR. The computations involving invariants were performed by using
Maple${}^{\text{\sc TM}}$ \cite{Maple} and Magma \cite{Magma}.

\appendix
\section{The polynomial $f(t)$} \label{dodatak}

The polynomial $f(t)$ defined in Sec. \ref{subsec:algorithm} can be written as
\begin{eqnarray}
f(t)=\sum_{k=0}^8 c_kt^k,
\end{eqnarray}
where
\begin{small}
\begin{align*}
c_0 ={}& 972(-64M_2M_4M_5^2M_6+16M_2^2M_4M_5^3+16M_2^2M_4^2M_5^2
+16M_2^3M_5^2M_6-8M_2^4M_4M_5^2 \displaybreak[0]\\
& -16M_2M_5^3M_6-16M_4^2M_5^3-8M_4M_5^4+64M_5^2M_6^2+3M_2^2M_5^4
-3M_2^4M_5^3+M_2^6M_5^2-M_5^5), \displaybreak[0]\\[0ex plus 2ex]
c_1 ={}& 1296(3M_1M_2^5M_5^2-9M_2^4M_3M_5^2+36M_1M_5^3M_6
-32M_3M_4M_5^3-64M_1M_2M_5M_6^2 \displaybreak[0]\\
& -M_1M_2^7M_5+4M_1M_2M_4M_5^3+128M_3M_5M_6^2
+48M_2^2M_3M_4M_5^2+32M_2^2M_3M_4^2M_5 \displaybreak[0]\\
& +48M_1M_4M_5^2M_6-16M_2^4M_3M_4M_5-48M_2M_3M_5^2M_6
+12M_1M_2^2M_5^2M_6 \displaybreak[0]\\
& +32M_2^3M_3M_5M_6-12M_1M_2^3M_4M_5^2+8M_1M_2^5M_4M_5
-16M_1M_2^4M_5M_6 \displaybreak[0]\\
& -16M_1M_2^3M_4^2M_5+64M_1M_2^2M_4M_5M_6-128M_2M_3M_4M_5M_6
+M_1M_2M_5^4 \displaybreak[0]\\
& +12M_2^2M_3M_5^3-3M_1M_2^3M_5^3-48M_3M_4^2M_5^2
+2M_2^6M_3M_5-5M_3M_5^4), \displaybreak[0]\\[0ex plus 2ex]
c_2 ={}& 108(256M_2^2M_3^2M_4^2+2736M_2M_4^2M_5^2
+252M_2M_4M_5^3-55296M_2M_4M_6^2  \displaybreak[0]\\
& +2124M_2^5M_4M_5-128M_2^4M_3^2M_4-6192M_2^3M_4^2M_5
+256M_2^3M_3^2M_6 \displaybreak[0]\\
& -13824M_2^4M_4M_6-10368M_4^2M_5M_6-768M_3^2M_4M_5^2
-4608M_2^3M_4^3-864M_2^7M_4 \displaybreak[0]\\
& -225M_2^7M_5+16M_2^6M_3^2+207M_2^5M_5^2+1024M_3^2M_6^2
-27M_2M_5^4+3456M_2^5M_4^2 \displaybreak[0]\\
& +1728M_2^6M_6-18M_1^2M_5^4-1224M_5^3M_6+13824M_2^3M_6^2
-27M_2^3M_5^3 \displaybreak[0]\\
& +4M_1^2M_2^8-160M_3^2M_5^3+1536M_1M_3M_4M_5M_6
-1920M_1^2M_2M_4M_5M_6 \displaybreak[0]\\
& +192M_1M_2M_3M_4M_5^2+384M_1M_2^2M_3M_5M_6
+1024M_1M_2^2M_3M_4M_6 \displaybreak[0]\\
& -384M_1M_2^3M_3M_4M_5+72M_2^9+36864M_6^3+768M_2^2M_3^2M_4M_5
-768M_2M_3^2M_5M_6 \displaybreak[0]\\
& +1728M_1M_3M_5^2M_6-16M_1M_2^7M_3-256M_1M_2^3M_3M_4^2
-144M_1M_2^3M_3M_5^2 \displaybreak[0]\\
& -1440M_1^2M_2M_5^2M_6-256M_1M_2^4M_3M_6+128M_1M_2^5M_3M_4
+480M_1^2M_2^2M_4^2M_5 \displaybreak[0]\\
& +96M_1M_2^5M_3M_5+14400M_2^2M_4M_5M_6
-1024M_1M_2M_3M_6^2+64M_1M_2M_3M_5^3 \displaybreak[0]\\
& +528M_1^2M_2^2M_4M_5^2+1872M_2^2M_5^2M_6-768M_3^2M_4^2M_5
+64M_1^2M_2^4M_4^2 \displaybreak[0]\\
& -1512M_2^3M_4M_5^2-288M_1^2M_4M_5^3+256M_1^2M_2^2M_6^2
+27648M_2^2M_4^2M_6 \displaybreak[0]\\
& +5184M_2M_4^3M_5-144M_2^4M_3^2M_5-9216M_2M_5M_6^2
-6336M_4M_5^2M_6 \displaybreak[0]\\
& +98M_1^2M_2^2M_5^3-2952M_2^4M_5M_6+6M_1^2M_2^6M_5
-32M_1^2M_2^6M_4+1152M_1^2M_5M_6^2 \displaybreak[0]\\
&  +64M_1^2M_2^5M_6-288M_1^2M_4^2M_5^2-90M_1^2M_2^4M_5^2
+288M_1^2M_2^3M_5M_6-144M_1^2M_2^4M_4M_5 \displaybreak[0]\\
& -256M_1^2M_2^3M_4M_6-1024M_2M_3^2M_4M_6+288M_2^2M_3^2M_5^2),\displaybreak[0]\\[0ex plus 2ex]
c_3 ={}& 36(-1024M_2M_3^3M_6+3528M_1M_2^6M_4+20736M_2M_3M_4^3
-81M_1M_2^4M_5^2 \displaybreak[0]\\
& -14688M_3M_5^2M_6-324M_2^3M_3M_5^2-2048M_3^3M_4M_5
+1404M_4M_5^3M_1 \displaybreak[0]\\
& -832M_1^3M_2^3M_4^2-172M_1^3M_2^3M_5^2-5616M_1M_4^2M_5^2
+768M_2^2M_3^3M_5-640M_1^3M_2^4M_6 \displaybreak[0]\\
& +82944M_1M_5M_6^2+8496M_2^5M_3M_4
+1656M_2^5M_3M_5+352M_1^3M_2^5M_4 \displaybreak[0]\\
& -41472M_3M_4^2M_6-24768M_2^3M_3M_4^2+4320M_1^3M_5^2M_6
+165888M_1M_4M_6^2 \displaybreak[0]\\
& -23040M_1M_2^2M_6^2-4464M_1M_2^5M_6-288M_1^2M_3M_5^3
+24M_1^2M_2^6M_3 \displaybreak[0]\\
& +31104M_1M_2^2M_4^3-18720M_1M_2^4M_4^2
-15552M_1M_4^3M_5-36864M_2M_3M_6^2 \displaybreak[0]\\
& +300M_1^3M_2^5M_5-432M_2M_3M_5^3-2304M_1^3M_2M_6^2
-11808M_2^4M_3M_6 \displaybreak[0]\\
& -567M_2^2M_5^3M_1+36M_1^3M_2M_5^3+4608M_1^2M_3M_6^2
-576M_1M_2^3M_3^2M_5 \displaybreak[0]\\
& -768M_1M_2^3M_3^2M_4-6372M_1M_2^4M_4M_5
-1184M_1^3M_2^3M_4M_5+54144M_1M_2^3M_4M_6 \displaybreak[0]\\
& +3024M_2M_3M_4M_5^2
-12096M_2^3M_3M_4M_5+6912M_1M_3^2M_5M_6-576M_1^2M_2^4M_3M_4 \displaybreak[0]\\
& -720M_1^2M_2^4M_3M_5+1176M_1^2M_2^2M_3M_5^2
+3072M_1^3M_2^2M_4M_6+3168M_1^3M_2^2M_5M_6 \displaybreak[0]\\
& +1152M_1^2M_2^3M_3M_6+57600M_2^2M_3M_4M_6
+21888M_2M_3M_4^2M_5+3072M_1M_3^2M_4M_6 \displaybreak[0]\\
& +11664M_1M_2M_5^2M_6
+864M_1M_2^2M_4M_5^2+384M_1M_2M_3^2M_5^2-145152M_1M_2M_4^2M_6 \displaybreak[0]\\
& +3456M_1^3M_4M_5M_6+19728M_1M_2^2M_4^2M_5-576M_1^3M_2M_4^2M_5
-50688M_3M_4M_5M_6 \displaybreak[0]\\
& +14976M_2^2M_3M_5M_6-2304M_1^2M_3M_4^2M_5
-3456M_1^2M_3M_4M_5^2+1920M_1^2M_2^2M_3M_4^2 \displaybreak[0]\\
& +768M_1M_2^2M_3^2M_6+12384M_1M_2^3M_5M_6-198M_1M_2^8
+243M_1M_5^4-640M_3^3M_5^2 \displaybreak[0]\\
& -36M_1^3M_2^7-1024M_3^3M_4^2-900M_2^7M_3-192M_2^4M_3^3
-7680M_1^2M_2M_3M_4M_6 \displaybreak[0]\\
& -11520M_1^2M_2M_3M_5M_6-54144M_1M_2M_4M_5M_6
+768M_1M_2M_3^2M_4M_5 \displaybreak[0]\\
& +4224M_1^2M_2^2M_3M_4M_5+192M_1M_2^5M_3^2+603M_1M_2^6M_5
+1024M_2^2M_3^3M_4), \displaybreak[0]\\[0ex plus 2ex]
c_4 ={}& 3(-8667M_2^8-324M_1^4M_5^3-1492992M_5M_6^2
+20736M_1^4M_6^2-444M_1^4M_2^6 \displaybreak[0]\\
& -16236M_1^2M_2^7-2592M_4^2M_5^2-11826M_2^4M_5^2
+13248M_2^5M_3^2-186624M_4^3M_5 \displaybreak[0]\\
& +311040M_2^2M_4^3-5120M_3^4M_5
-8192M_3^4M_4-3981312M_4M_6^2+331776M_2^2M_6^2 \displaybreak[0]\\
& -189216M_2^4M_4^2-10368M_2^5M_6+17820M_2^6M_5+71280M_2^6M_4
+3072M_2^2M_3^4 \displaybreak[0]\\
& -5504M_1^3M_2^3M_3M_5-2187M_5^4+29376M_1^4M_2^2M_4M_5
-18944M_1^3M_2^3M_3M_4 \displaybreak[0]\\
& -48384M_1^4M_2M_4M_6-95040M_1^4M_2M_5M_6
+198144M_1M_2^3M_3M_6-2592M_1M_2^4M_3M_5 \displaybreak[0]\\
& -101952M_1M_2^4M_3M_4-20736M_1^2M_4M_5M_6
-92160M_1^2M_2M_3^2M_6-55296M_1^2M_3^2M_4M_5 \displaybreak[0]\\
& +33792M_1^2M_2^2M_3^2M_4+18816M_1^2M_2^2M_3^2M_5
+219456M_1^2M_2M_4^2M_5-112320M_1^2M_2^2M_5M_6 \displaybreak[0]\\
& +315648M_1M_2^2M_3M_4^2+4096M_1M_2M_3^3M_4
+1728M_1^3M_2M_3M_5^2+138240M_1^3M_3M_5M_6 \displaybreak[0]\\
& +4096M_1M_2M_3^3M_5-9216M_1^3M_2M_3M_4^2
+663552M_2M_4M_5M_6+55296M_1^3M_3M_4M_6 \displaybreak[0]\\
& +1052928M_1^2M_2^2M_4M_6-27216M_2^2M_5^2M_1M_3
+16848M_4M_5^2M_1^2M_2+67392M_4M_5^2M_1M_3 \displaybreak[0]\\
& +48384M_2M_3^2M_4M_5-179712M_1M_3M_4^2M_5
-95904M_1^2M_2^3M_4M_5+50688M_1^3M_2^2M_3M_6 \displaybreak[0]\\
& -5760M_1^2M_2^4M_3^2+15552M_1M_5^3M_3
+120528M_1^2M_2^5M_4-186624M_1^2M_2M_4^3 \displaybreak[0]\\
& -1492992M_1^2M_2M_6^2-175680M_1^2M_2^3M_4^2
-6912M_1^2M_3^2M_5^2-297504M_1^2M_2^4M_6 \displaybreak[0]\\
& -12960M_1^4M_4M_5^2+53136M_2^2M_5^2M_4-324M_1^2M_2M_5^3
-1568M_1^4M_2^4M_4 \displaybreak[0]\\
& -5708M_1^4M_2^4M_5-1728M_1^4M_2^3M_6
-5184M_1^4M_4^2M_5-10368M_2M_3^2M_5^2 \displaybreak[0]\\
& -209952M_6M_1^2M_5^2+4800M_1^3M_2^5M_3-82944M_2^3M_5M_6
+4428M_1^4M_2^2M_5^2 \displaybreak[0]\\
& +1327104M_1M_3M_6^2-155520M_2M_5^2M_6-96768M_2^3M_3^2M_4
-5184M_2^3M_3^2M_5 \displaybreak[0]\\
& +2488320M_2M_4^2M_6-405504M_3^2M_4M_6
-235008M_3^2M_5M_6+119808M_2^2M_3^2M_6 \displaybreak[0]\\
& -112752M_2^4M_4M_5-248832M_1M_3M_4^3+175104M_2M_3^2M_4^2+212544M_2^2M_4^2M_5 \displaybreak[0]\\
& -18432M_1^2M_3^2M_4^2+622080M_1^2M_4^2M_6-3072M_1M_2^3M_3^3
+9648M_1M_2^6M_3 \displaybreak[0]\\
& +36864M_1M_3^3M_6+14400M_1^4M_2^2M_4^2+2268M_2^3M_5^2M_1^2
+11988M_1^2M_2^5M_5 \displaybreak[0]\\
& +373248M_1M_2M_3M_5M_6+27648M_1M_2^2M_3M_4M_5
-866304M_1M_2M_3M_4M_6 \displaybreak[0]\\
& -580608M_2^3M_4M_6-11664M_4M_5^3-559872M_4^4+4860M_2^2M_5^3),\displaybreak[0]\\[0ex plus 2ex]
c_5 ={}& 4(-1024M_3^5+632448M_1M_2^3M_4^2+11988M_1^2M_2^5M_3
-23652M_2^4M_3M_5 \displaybreak[0]\\
& -6912M_2M_3^3M_5+16128M_2M_3^3M_4
-159408M_1^3M_4^2M_5-82944M_2^3M_3M_6 \displaybreak[0]\\
& +69120M_1^3M_3^2M_6
-18432M_1^2M_3^3M_4-4608M_1^2M_3^3M_5+6272M_1^2M_2^2M_3^3 \displaybreak[0]\\
& +1024M_1M_2M_3^4-5184M_1^5M_2M_4^2-1296M_1M_2^4M_3^2
+186624M_1M_2M_4^3 \displaybreak[0]\\
& -89856M_1M_3^2M_4^2+1029024M_1M_2^4M_6
-13365M_2^2M_5^2M_1^3+19440M_2^3M_5^2M_1 \displaybreak[0]\\
& -34992M_4M_5^2M_3
+14580M_2^2M_5^2M_3-5184M_1^4M_3M_4^2+15552M_1^5M_4M_6 \displaybreak[0]\\
& +42768M_1^5M_5M_6-185940M_1^3M_2^4M_4-2752M_1^3M_2^3M_3^2
-5184M_4^2M_5M_3 \displaybreak[0]\\
& +349920M_6M_1M_5^2+5225472M_1M_2M_6^2
+373248M_1M_4^2M_6+576M_1^5M_2^3M_5 \displaybreak[0]\\
& +26352M_1^5M_2^2M_6-8496M_1^5M_2^3M_4-5708M_1^4M_2^4M_3
+384768M_1^3M_2^3M_6 \displaybreak[0]\\
& +375408M_1^3M_2^2M_4^2+39852M_1^3M_5^2M_4-972M_1^4M_5^2M_3
+23328M_1M_5^2M_3^2 \displaybreak[0]\\
& -59616M_1M_2^5M_5-112752M_2^4M_3M_4
+2241M_1^3M_5M_2^4+40176M_1M_2^7 \displaybreak[0]\\
& +212544M_2^2M_3M_4^2+17820M_2^6M_3-419904M_6M_1^2M_3M_5
+4536M_1^2M_2^3M_3M_5 \displaybreak[0]\\
& -972M_1^2M_2M_5^2M_3-279936M_6M_2^2M_1M_5+248832M_6M_1^3M_2M_5
+13824M_1M_2^2M_3^2M_4 \displaybreak[0]\\
& -330480M_1M_2^5M_4-186624M_3M_4^3
-1492992M_3M_6^2-78336M_3^3M_6-15552M_1^3M_4^3 \displaybreak[0]\\
& +22977M_1^3M_2^6+1119744M_1^3M_6^2+243M_1^3M_5^3
-1728M_2^3M_3^3+2240M_1^5M_2^5 \displaybreak[0]\\
& -8748M_5^3M_3+33696M_1^2M_2M_3M_4M_5+29376M_1^4M_2^2M_3M_4
+8856M_1^4M_2^2M_3M_5 \displaybreak[0]\\
& -95040M_1^4M_2M_3M_6-25920M_1^4M_3M_4M_5
-20736M_1^2M_3M_4M_6 \displaybreak[0]\\
& +219456M_1^2M_2M_3M_4^2-311040M_2M_3M_5M_6
-6480M_1^5M_2M_4M_5 \displaybreak[0]\\
& +663552M_2M_3M_4M_6-95904M_1^2M_2^3M_3M_4
-112320M_1^2M_2^2M_3M_6 \displaybreak[0]\\
& +1728M_1^3M_2M_3^2M_5-27216M_1M_2^2M_3^2M_5
+106272M_2^2M_4M_5M_3 \displaybreak[0]\\
& +386208M_2^3M_5M_4M_1
-97200M_4M_5^2M_1M_2+67392M_1M_3^2M_4M_5 \displaybreak[0]\\
& +186624M_1M_2M_3^2M_6-4375296M_1M_2^2M_4M_6
-1451520M_1^3M_2M_4M_6\displaybreak[0]\\
& +746496M_6M_1M_4M_5 -611712M_1M_2M_4^2M_5
+18792M_1^3M_2^2M_4M_5),\displaybreak[0]\\[0ex plus 2ex]
c_6 ={}& 48(64M_1^3M_2M_3^3+1044M_1^6M_2^2M_4+44064M_1^4M_4M_6
+252M_2^3M_3^2M_1^2 \displaybreak[0]\\
& -1008M_2^2M_3^3M_1+6924M_1^4M_4M_2^3
-108M_1^4M_5M_3^2+249M_1^3M_2^4M_3 \displaybreak[0]\\
& +1728M_1M_5M_3^3+81M_1^3M_5^2M_3
-972M_1^4M_2^3M_5+1620M_2^2M_5M_3^2 \displaybreak[0]\\
& +810M_2^2M_5^2M_1^2+2496M_1M_3^3M_4
+3996M_1^2M_2^4M_5+5904M_2^2M_3^2M_4 \displaybreak[0]\\
& -6624M_1M_2^5M_3-84528M_1^2M_2^2M_4^2-23328M_6M_1^2M_3^2
+108M_1^6M_5M_2^2 \displaybreak[0]\\
& -17280M_2M_3^2M_6-3888M_4M_5M_3^2-324M_1^6M_4M_5
-17712M_1^3M_3M_4^2 \displaybreak[0]\\
& +64M_1^5M_2^3M_3-3024M_1^6M_2M_6
-15120M_6M_1^4M_2^2-15552M_6M_1^4M_5 \displaybreak[0]\\
& -142128M_6M_2^3M_1^2+285120M_2^2M_4M_6
+19440M_6M_2^2M_5-46656M_6M_4M_5 \displaybreak[0]\\
& +25920M_4^2M_5M_2+45360M_1^2M_4^2M_5-1440M_1^4M_3^2M_4
-12528M_1^4M_2M_4^2 \displaybreak[0]\\
& +492M_1^4M_2^2M_3^2+4752M_1^5M_3M_6+5832M_2M_4M_5^2
-1296M_2^7+11016M_2^5M_4 \displaybreak[0]\\
& -1314M_2^4M_3^2-15552M_2^3M_4^2
+42912M_1M_2^3M_3M_4-16848M_2^3M_5M_4 \displaybreak[0]\\
& -1296M_2^3M_5^2-2430M_4M_5^2M_1^2
-31104M_4^2M_6-17496M_6M_5^2-373248M_2M_6^2 \displaybreak[0]\\
& -236M_1^6M_2^4-128M_1^2M_3^4-38880M_1^2M_4^3-900M_2^5M_1^4
+2592M_2^5M_5 \displaybreak[0]\\
& -1458M_5^2M_3^2-653184M_1^2M_6^2-64152M_2^4M_6
-192M_2M_3^4 \displaybreak[0]\\
& -25812M_2^2M_4M_5M_1^2+4320M_2^3M_5M_3M_1
-108M_1^2M_5M_3^2M_2-67968M_1M_2M_4^2M_3 \displaybreak[0]\\
& +2088M_1^3M_2^2M_4M_3+77760M_6M_1M_5M_3-11664M_6M_1^2M_5M_2 \displaybreak[0]\\
& -720M_1^5M_2M_3M_4
-31104M_6M_2^2M_1M_3-2970M_1^3M_2^2M_5M_3+27648M_6M_1^3M_2M_3 \displaybreak[0]\\
& +575424M_6M_1^2M_4M_2+82944M_6M_1M_4M_3
+1872M_1^2M_3^2M_4M_2 \displaybreak[0]\\
& +4644M_1^4M_5M_4M_2+8856M_1^3M_4M_5M_3-21600M_1M_3M_4M_5M_2
-288M_3^2M_4^2 \displaybreak[0]\\
& -31104M_2M_4^3-5958M_1^2M_2^6+47538M_1^2M_2^4M_4),\displaybreak[0]\\[0ex plus 2ex]
c_7 ={}& 576(-288M_2^3M_3M_5+1728M_1M_2^2M_4M_5+1296M_2M_3M_4M_5
+60M_2^2M_3^3 \displaybreak[0]\\
& -1200M_1M_2M_3^2M_4-288M_1M_2^4M_5-2868M_1^2M_2^2M_3M_4
+48M_1M_3^4+8M_1^7M_2^3 \displaybreak[0]\\
& -108M_3^3M_5+5M_2^4M_1^5+108M_1^7M_6
-108M_1^5M_4^2+93312M_1M_6^2-4M_1^4M_3^3 \displaybreak[0]\\
& +240M_1M_2^3M_3^2-2304M_1M_2^4M_4-2028M_1^3M_2^3M_4
-1872M_2^3M_3M_4+444M_1^2M_2^4M_3 \displaybreak[0]\\
& +2880M_2M_3M_4^2+2016M_1M_2^2M_4^2
+3744M_1^3M_2M_4^2+5040M_1^2M_3M_4^2 \displaybreak[0]\\
& -3M_1^5M_4M_2^2+972M_6M_1^3M_5-11664M_6M_1^3M_4
-108M_1^4M_2^3M_3+72M_1^3M_2^3M_5 \displaybreak[0]\\
& -9M_1^5M_2^2M_5+12M_1^6M_2^2M_3-4M_1^2M_2M_3^3
+492M_1^3M_3^2M_4+27M_1^5M_5M_4 \displaybreak[0]\\
& -36M_1^7M_2M_4-36M_1^6M_3M_4-1728M_6M_1^4M_3+3780M_6M_1^3M_2^2+216M_6M_1^5M_2 \displaybreak[0]\\
& +15552M_6M_2^3M_1-3888M_3M_6M_5+4320M_6M_1M_3^2
-2592M_1M_4^2M_5-5184M_6M_4M_3 \displaybreak[0]\\
& +2160M_6M_2^2M_3-165M_2^2M_3^2M_1^3+9M_1^3M_5M_3^2
-144M_3^3M_4-324M_1^3M_4M_5M_2 \displaybreak[0]\\
& +180M_2^2M_3M_5M_1^2+516M_1^4M_2M_4M_3-67392M_6M_1M_4M_2
-1296M_6M_1^2M_3M_2 \displaybreak[0]\\
& -540M_1^2M_4M_5M_3+288M_1M_2^6+264M_1^3M_2^5
+288M_2^5M_3+10368M_1M_4^3), \displaybreak[0]\\[0ex plus 2ex]
c_8 ={}& 2304(-96M_1M_2^4M_3-48M_2^3M_3^2-1728M_4^3
+576M_1M_2^2M_3M_4-648M_3^2M_6 \displaybreak[0]\\
& +M_1^3M_3^3-27M_1^6M_6-2M_1^6M_2^3-9M_3^4
+432M_2^2M_4^2+24M_1^3M_2^3M_3-90M_1^2M_3^2M_4 \displaybreak[0]\\
& +30M_1^2M_2^2M_3^2-648M_1^2M_2M_4^2-11664M_6^2
-18M_1^4M_2^2M_4-3M_1^5M_2^2M_3 \displaybreak[0]\\
& +324M_1^3M_3M_6+7776M_2M_4M_6+9M_1^6M_2M_4
+9M_1^5M_3M_4-648M_1^2M_2^2M_6 \displaybreak[0]\\
& -864M_1M_3M_4^2-108M_1^3M_2M_3M_4
-1728M_6M_2^3+27M_1^4M_4^2+3M_1^4M_2^4 \displaybreak[0]\\
& +1944M_6M_1^2M_4
+360M_1^2M_2^3M_4+216M_2M_3^2M_4-48M_1^2M_2^5).
\end{align*}
\end{small}


\end{document}